\documentclass[12pt]{article}
\usepackage{amsfonts}
\usepackage{amssymb,amsmath,amsthm,latexsym}
\usepackage[numbers,compress]{natbib}
\usepackage{amsmath}
\usepackage{cases}
\usepackage{indentfirst}
\usepackage{mathrsfs}
\usepackage{amsfonts}
\usepackage{amsmath}
\usepackage{amsthm}
\usepackage[usenames]{color}
\usepackage{footnote}
\usepackage{longtable}
\usepackage{hyperref}
\usepackage[all]{xy}
\usepackage{amsmath,amscd}
\usepackage{multirow}
\usepackage{booktabs}
\usepackage{clrscode}
\usepackage{listings}
\usepackage{empheq}
\usepackage{color}
\usepackage{array}

\allowdisplaybreaks[4]
\newtheorem{theorem}{Theorem}[section]

\newtheorem{corollary}[theorem]{Corollary}
\newtheorem{lemma}[theorem]{Lemma}

\newtheorem{define}[theorem]{Definition}

\numberwithin{equation}{section}

\usepackage{color}
\usepackage{graphicx}
\usepackage{url}
\usepackage{booktabs} 
\usepackage{placeins} 

\def\|{|\kern-1.25pt|}

\def\bfs#1{{\setbox0=\hbox{$\scriptstyle#1$} 
     \kern-.020em\copy0\kern-\wd0
     \kern .040em\copy0\kern-\wd0
     \kern-.020em\raise.01em\box0 }}

\newcommand{\rmv}[1]{}
\usepackage{transparent}
\usepackage{tabularx}
\usepackage{float}
\usepackage{amsmath,amscd}

\begin{document}

\title{New results of $0$-APN power functions over $\mathbb{F}_{2^n}$}

\author{Yan-Ping Wang\textsuperscript{1}   
  \and
        Zhengbang Zha\textsuperscript{2}   {\thanks{Corresponding author.  Email addresses:~zhazhengbang@163.com}}     \and
        \vspace*{0.0cm}\\
{\small $^1$ College of Mathematics and Statistics, Northwest Normal University, Lanzhou 730070, China}\\
{\small $^2$ School of Mathematical Sciences, Luoyang Normal University, Luoyang 471934, China} \\
}

\date{}
\maketitle
\begin{abstract}
Partially APN functions attract researchers' particular interest recently. It plays an important role in studying APN functions.
In this paper, based on the multivariate method and resultant elimination, we propose several new infinite classes of $0$-APN power functions over $\mathbb{F}_{2^n}$. Furthermore, two infinite classes of $0$-APN power functions $x^d$ over $\mathbb{F}_{2^n}$ are characterized completely where  $(2^k-1)d\equiv 2^m-1~({\rm mod}\ 2^n-1)$ or $(2^k+1)d\equiv 2^m+1~({\rm mod}\ 2^n-1)$ for some positive integers $n, m, k$. These infinite classes of $0$-APN power functions can explain some examples of exponents of Table $1$ in \cite{BKRS2020}.
\end{abstract}




\section{Introduction}
\label{intro}

Let $\mathbb{F}_{2^n}$ be the finite field with $2^{n}$ elements, and $\mathbb{F}_{2^{n}}^{*}$ be the multiplicative group which consists of all the nonzero elements of $\mathbb{F}_{2^{n}}$. For a function $f:\mathbb{F}_{2^{n}}\rightarrow\mathbb{F}_{2^{n}}$, the derivative of $f(x)$ is defined by $D_{a}f(x) = f(x+a) + f(x)$, where $x\in\mathbb{F}_{2^n}$ and $a\in\mathbb{F}_{2^{n}}^{*}$.
For any $b\in\mathbb{F}_{2^{n}}$, we define $\delta_{f}(a,b) = \#\{x\in\mathbb{F}_{2^{n}} ~|~ D_{a}f(x)=b\} $.
The maximum value $\delta_{f} = \max\limits_{a\in\mathbb{F}_{2^{n}}^{*},b\in\mathbb{F}_{2^{n}}}\delta_{f}(a,b)$ is called the differential uniformity of $f(x)$.
A function $f(x)$ over $\mathbb{F}_{2^{n}}$ is called Almost Perfect Nonlinear (APN for short) if its differential uniformity $\delta_{f}$ equals $2$.

Block ciphers should be designed to resist all classical attacks. S-boxes are the core components of block ciphers. The primary purpose of S-boxes is to produce confusion inside block ciphers.
Such S-boxes are nonlinear functions over $\mathbb{F}_{2^{n}}$. These functions should have low differential uniformity for resisting differential attacks \cite{BS1991}.
In cryptography, APN functions over $\mathbb{F}_{2^{n}}$ are the optimal ones to resist differential attacks.
Therefore,  APN functions have important applications in block ciphers.
For example, the APN functions $x^{81}$ over $\mathbb{F}_{2^9}$ and $x^5$ over $\mathbb{F}_{2^7}$ have been respectively used in  MISTY and KASUMI block ciphers. APN power functions over finite fields attract researchers' interest for their simple algebraic form and some additional properties. So far, researchers only found six classes of APN power functions on $\mathbb{F}_{2^n}$: Gold functions~\cite{GR1968}, Kasami functions~\cite{KT1971}, Welch functions \cite{DobbWelch1999}, Niho functions \cite{DobbNiho1999}, Inverse functions~\cite{NK1994} and Dobbertin functions~\cite{Hans2001}. Furthermore, many results of APN functions appeared in the recent years, the reader may refer to \cite{CC2021} and references therein for more information.

In the conference SETA 2018 \cite{BKKRS2018} (for the journal edition, see \cite{BKKRS2020}), to study the conjecture of the highest possible algebraic degree of APN functions, Budaghyan et al. firstly proposed the following concept of partially APN.
\begin{define}\emph{({\rm\cite{BKKRS2020}})}\label{p-APN}
Let $f(x)$ be a function from $\mathbb{F}_{2^{n}}$ to itself. For a fixed $x_{0}\in\mathbb{F}_{2^{n}}$, then we call the function $f(x)$ is $x_{0}$-APN  (or partially APN) if all the points $x, y$ satisfying $f(x_{0}) + f(x) + f(y) + f(x_{0} + x + y) = 0$ belong to the curve $(x_{0} + x)(x_{0} + y)(x + y) = 0$.
\end{define}
It is obvious that $f(x)$ is $x_{0}$-APN for any $x_{0}\in\mathbb{F}_{2^{n}}$ if $f(x)$ is an APN function.
Partially APN power functions are of APN-like exponents that are not APN. 
In \cite{BKKRS2020} the authors provided some characterizations and propositions of partially APN functions.
Furthermore, they presented several classes of $x_{0}$-APN functions for some $x_{0}\in\mathbb{F}_{2^{n}}$. 
Pott \cite{Pott2019} pointed out that for any $n\geq3$ there exist partial $0$-APN permutations on $\mathbb{F}_{2^n}$, when he studied the relation between Steiner triple systems and partial $0$-APN permutations. 
Based on the idea and the known instances of $x_{0}$-APN functions, it is very interesting to construct more infinite classes of $x_{0}$-APN functions for some $x_{0}\in\mathbb{F}_{2^{n}}$.
If $f$ is a power function, then we only need to consider the partial APN properties of $f$ at $0$ or $1$ since it's special algebraic structure \cite{BKKRS2020}. Moreover, $f$  is $0$-APN if and only if the equation $f(x + 1) + f(x) + 1 = 0$ has no solution in $\mathbb{F}_{2^n}\backslash\{0,1\}$. In \cite{BKKRS2020,BKRS2020}, Budaghyan et al. explicitly constructed a number of $0$-APN but not APN power functions $f(x) = x^d$ over $\mathbb{F}_{2^{n}}$. They further listed all power functions over $\mathbb{F}_{2^{n}}$ for $1\leq n\leq11$ that are $0$-APN but not APN in Table $1$ of \cite{BKRS2020}. Recently, Qu and Li \cite{QL2022} constructed seven classes of $0$-APN power functions over $\mathbb{F}_{2^n}$. Two of them were proved to be locally-APN \cite{BCC2011}.
The purpose of this paper is to construct more new classes of $0$-APN power functions over $\mathbb{F}_{2^n}$.
As a result, we propose several new infinite classes of $0$-APN power functions over $\mathbb{F}_{2^n}$ by using the multivariate method and resultant elimination. For the sake of clarity and further studying, we list the new results of this paper in Table \ref{tab}, along with some prior ones.

The paper is organized as follows. Section \ref{pre} gives some necessary definitions and results.
In Section \ref{third}, several infinite classes of $0$-APN power functions are presented over different finite fields.
In Section \ref{fourth}, we completely characterize two infinite classes of the $0$-APN power functions $x^d$ over $\mathbb{F}_{2^n}$
under the conditions of $(2^k-1)d\equiv 2^m-1~({\rm mod}\ 2^n-1)$ and $(2^k+1)d\equiv 2^m+1~({\rm mod}\ 2^n-1)$ respectively.
The conclusion is given in Section \ref{conclu}.

\begin{table}[H]
\scriptsize
\centering
\caption{Known results of $0$-APN power functions $f(x)=x^d$ over $\mathbb{F}_{2^n}$}
\label{tab}
\begin{tabular}{m{1.2cm}p{2.2cm}p{5.2cm}m{1.6cm}p{5.2cm}}
\hline\noalign{\smallskip}
 Number  &           $x^d$       &     Conditions                &   References            &  Values of $(d, n)(n\leq11)$        \\
\noalign{\smallskip}\hline\noalign{\smallskip}
1)       &  $ x^{2^i-1}$                  & $\gcd(n, i-1)=1$          & \cite{BKKRS2020}        & $(7, 7), (31, 7), (15, 8),(63, 8)$, $(7, 9), (63, 9),(15, 10), (255, 10)$, $(7, 11), (15, 11), (31, 11), (127, 11)$, $ (255, 11), (511, 11)$   \\

2)       &  $ x^{21}$                     &  $6\nmid n$            & \cite{BKRS2020}        &   $(21, 7), (21, 8), (21, 9),(21, 10)$, $(21, 11)$           \\

3)       &  $ x^{2^r+2^t-1}$              & $\gcd(n, r)=\gcd(n, t)=1$  & \cite{BKRS2020}        &   $(7, 7), (31, 7), (47, 7),(47, 7)$, $ (15, 8), (63, 8), (7, 9), (63, 9)$, $(15, 10),(255, 10), (7, 11), (11, 11)$, $(15, 11), (19, 11),(23, 11),(31, 11)$, $(39, 11),  (47, 11),(67, 11),(71, 11)$, $(79, 11),(127, 11),(159, 11)$, $ (191, 11), (255, 11),(319, 11)$, $(383, 11),(511, 11),(767, 11)$  \\

4)       &  $ x^{2^{2t}+2^{t}+1}$         &  $n=4t$, $t$ even                                       & \cite{BKRS2020}        &   $(21, 8)  $       \\

5)       &  $ x^{2^n-2^s}$                &  $\gcd(n, s+1)=1$                                  & \cite{BKRS2020}        &   $(7, 7),(31, 7),(15, 8),(63, 8)$,  $(7, 9),(63, 9),(15, 10),( 255, 10)$,    $(7, 11),(15, 11),(31, 11),( 127, 11)$, $(255, 11),(511, 11)$       \\

6)       &  $ x^{j\cdot(2^{m}-1)}$       &  $n=2m$, $\gcd(j, 2^m+1)=1$, $m$ even, $j$ integer    & \cite{QL2022}        &   $ (15, 8),(45, 8)$    \\

7)       &  $ x^{2^{2m-1}-2^{m}-1}$       &  $n=2m$, $m$ even, $3\nmid m$                       & \cite{QL2022}        &   $ (111, 8)$    \\  

8)       &  $ x^{2^{2m-1}-2^{m-1}-1}$     &  $n=2m$, $m$ odd                                     & \cite{QL2022}        &   $(27, 6), (495, 10) $     \\

9)       &  $ x^{2^{3k}-2^{2k}+2^k-1}$    &  $n=4k$, $k$ even                                   & \cite{QL2022}        &   $ (51, 8)$            \\

10)        &  $ x^{2^{2m}-2^{m}-1} $       &  $n=2m+1$, $m\not\equiv 1~({\rm mod}\ 3)$              & \cite{QL2022}          &   $ (55, 7), (991, 11)$   \\

11)       &  $ x^{2^{2m-1}-2^{m-1}-1} $   &  $n=2m+1$, $m$ integer                              & \cite{QL2022}          &   $(111, 9), (479, 11) $   \\

12)       &  $ x^{2^{2m-1}-2^{m}-1} $     &  $n=2m+1$, $m$ integer                                & \cite{QL2022}          &   $(119, 9), (495, 11)  $   \\
13)       &  $ x^{2^{2k-1}-2^{k-1}-1}$    &  $n=4k$, $k$ odd                                    & Thm \ref{th1.1}        &   $ --$       \\ 

14)       &  $ x^{2^{2k-1}+2^{k}+1}$      &  $n=2k+1$, $k$ integer                                & Thm \ref{th:2.1}     & $(21, 7), (41, 9), (81, 11)$  \\  

15)       &  $ x^{2^{2k}+2^{k+1}+1} $     &  $n=2k+1$, $k\not\equiv 1~({\rm mod}\ 3)$              & Thm \ref{th:2.2}        &   $(49, 11)$     \\ 

16)       &  $ x^{2^{k+1}-2^{k-1}-1} $    &  $n=2k+1$, $k\not\equiv 1~({\rm mod}\ 3)$              & Thm \ref{th:2.3}        &   $(47, 11) $     \\

17)       &  $ x^{2^{2k}-2^{k+1}-1}$      &  $n=2k+1$, $k\not\equiv 4~({\rm mod}\ 9)$             & Thm \ref{th:2.4}        &   $(47, 7), (959, 11) $  \\   

18)       &  $ x^{2^{2k}+2^{k+1}+1}$      &  $n=3k-1$, $k$ integer                                 & Thm \ref{th:2.5}        &   $ (21, 8), (73, 11)$     \\  

19)       &  $ x^{2^{2k+1}+2^{k+1}+1}$    &  $n=3k-1$, $k$ even                                   & Thm \ref{th:2.6}        &   $ (81, 11) $  \\  

20)       &  $ x^{2^{2k+1}+2^{k}+1}$      &  $n=3k-1$, $k$ even                                      & Thm \ref{th:2.7}        &   $ (69, 11) $     \\ 


21)       &  $ x^{3\cdot2^{2k}+1}$        &  $n=3k-1$, $k$ even                                       & Thm \ref{th:2.16}        &   $(11, 11)  $    \\  

22)       &  $ x^{2^{2k-1}-2^{k}-1}$      &  $n=3k$, $k\not\equiv 0~({\rm mod}\ 3)$                  & Thm \ref{th:2.9}        &   $-- $  \\ 

23)       &  $ x^{2^{2k-1}+2^{k}+1}$      &  $n=3k$, $k$ odd                                         & Thm \ref{th:2.10}        &   $(41, 9)  $     \\ 

24)       &  $ x^{2^{2k}-2^{k+1}-1}$      &  $n=3k$, $k$ odd                                         & Thm \ref{th:2.11}        &   $(21, 9) $      \\ 

25)       &  $ x^{2^{2k+1}-2^{k}-1} $     &  $n=3k$, $k$ integer                                     & Thm \ref{th:2.18}        &   $(27, 6), (119, 9) $ \\ 


26)       &  $ x^{3\cdot(2^{k+1}-1)}$     &  $n=3k+1$, $k\not\equiv 11~({\rm mod}\ 34)$             & Thm \ref{th:2.20}        &   $(21, 7), (45, 10) $ \\  

27)       &  $ x^{d}$     &  $\gcd(n, mk)=\gcd(n, m-k)=1$, $(2^k-1)d\equiv 2^m-1~({\rm mod}\ 2^n-1)$, $n, m, k$ integers & Thm \ref{th:4.1}     & 
$(7, 7),(19, 7), (21, 7),(31, 7)$,  $(47, 7),(55, 7),(341, 11),(731, 11)$, $(887, 11),(991, 11),(137, 11),(293, 11)$, $(511, 11),(73, 11),(307, 11),(99, 11)$, $(463, 11),(879, 11),(255, 11),(85, 11)$, $(199, 11),(959, 11),(767, 11),(495, 11)$, $(67, 11),(443, 11),(703, 11),(895, 11)$, $(153, 11),(301, 11),(687, 11)$   \\

28)       &  $ x^{d}$     &  $(2^k+1)d\equiv 2^m+1~({\rm mod}\ 2^n-1)$, $\gcd(n, m+k)=\gcd(n, m-k)=1$, $\frac{n}{\gcd(n, k)}$ odd, $n, m, k$ integers & Thm \ref{th:4.2}(i)     &   $(27, 6), (207, 10),(231, 10),(189, 10)$,  $(363, 10),(11, 11),(121, 11),(171, 11)$, $(423, 11),(205, 11),(235, 11),(343, 11)$, $(429, 11),(221, 11),(189, 11)$\\
          &               &  $(2^k+1)d\equiv 2^m+1~({\rm mod}\ 2^n-1)$, $d\equiv 0~({\rm mod}\ 3)$, $\frac{n}{\gcd(n, k)}$ even, $\gcd(k, n)=1$, $\gcd(m+k,n)= \gcd(m-k,n)=2$, $n$ even, $k$ and $m$ odd    &   Thm \ref{th:4.2}(ii) & $(231, 10)$\\
\noalign{\smallskip}\hline
\end{tabular}
\footnotesize
\raggedright
The last column $(d, n)$ denotes the examples of $0$-APN (but not APN) power functions $x^d$ over $\mathbb{F}_{2^n}(n\leq11)$, which appeared in Table $1$ in \cite{BKRS2020}.
Note that the functions 13) and 22) only exist such examples over $\mathbb{F}_{2^n}$ of $n\geq12$, thus the corresponding examples do not list.
\end{table}

\section{Preliminaries}\label{pre}

In this section, we give some necessary definitions and results which will be used in this paper. 

We firstly recall the CCZ equivalence of power functions on $\mathbb{F}_{p^n}$.
\begin{lemma}\label{Lem:CCZ}\emph{({\rm\cite[Theorem $1$]{ Demp2018}})}
The power functions $p_{k}(x)=x^k$ and $p_{l}(x)=x^l$ on $\mathbb{F}_{p^n}$ are CCZ equivalent, if and only if
there exists a positive integer $0\leq a < n$, such that $l\equiv p^a k~({\rm mod}\ p^n-1)$  or $kl\equiv p^a~({\rm mod}\ p^n-1)$.
\end{lemma}

To investigate the solutions of a system of polynomial equations, the resultant of two polynomials is needed.

\begin{define}\emph{({\rm \cite[p.36]{LN97}})}
Let $q=p^r$, where $p$ is a prime and $r$ is a positive integer. Let $f(x)=a_{0}x^{n}+a_{1}x^{n-1}+\cdots +a_{n}\in \mathbb{F}_q[x]$ and $g(x)=b_{0}x^{m}+b_{1}x^{m-1}+\cdots+ b_{m}\in \mathbb{F}_q[x]$ be
two polynomials of  degree $n$ and $m$ respectively,  where $n, m \in \mathbb{N}$. Then the resultant $Res(f, g)$ of the two polynomials is defined by the determinant
\begin{eqnarray*}
Res(f, g)=
\left|\begin{array}{cccccccc}
   a_{0} &    a_{1}    & \cdots &   a_{n}  & 0  &    &\cdots & 0   \\
    0  &  a_{0} &    a_{1}    & \cdots &   a_{n}  & 0  & \cdots & 0 \\
    \vdots  &    &        &      &     &      &          &  \vdots    \\
    0  &  \cdots &   0   &  a_{0} &    a_{1}  &     &  \cdots &   a_{n} \\
    b_{0} &    b_{1}    & \cdots &      & b_{m}  & 0  & \cdots &  0     \\
      0   &    b_{0} &   b_{1}  & \cdots &    &  b_{m}  & \cdots &  0 \\
    \vdots &          &      &     &      &     &    &\vdots       \\
    0  &   \cdots   & 0 &  b_{0} &    b_{1} &    & \cdots &  b_{m}  \\
\end{array}\right|
& \begin{array}{l}
\left.\rule{0mm}{9.80mm}\right\}$m$~rows \\
\\\left.\rule{0mm}{9.80mm}\right\}$n$~ rows
\end{array}\\[0pt]
\end{eqnarray*}
of order $m+n$.
\end{define}
If the degree of $f$ is $Deg(f) = n$  (i.e., $a_{0}\neq 0$) and $f(x)=a_{0}(x -\alpha_{1})(x -\alpha_{2})\cdots (x -\alpha_{n})$ in
the splitting field of $f$ over $\mathbb{F}_q$, then $Res(f, g)$ is also given by the formula
\begin{eqnarray*}
Res(f, g)=a_{0}^{m}\prod_{i=1}^n g(\alpha_{i}).
\end{eqnarray*}
In this case, we have $Res(f, g) = 0$ if and only if $f$ and $g$ have a common root, which means that $f$ and $g$ have a common
divisor in $\mathbb{F}_q[x]$ of positive degree.

For  two polynomials $F(x,y),\, G(x,y)\in \mathbb{F}_q[x,y]$ of positive degree in $y$, the resultant $Res(F,G,y)$ of $F$ and $G$ with respect to $y$ is the resultant of $F$ and $G$ when considered as polynomials in the single variable $y$. 
In this case, $Res(F,G,y)\in \mathbb{F}_q[x] \cap \langle F,G\rangle$, where $\langle F,G\rangle$ is the ideal generated by $F$ and $G$. Thus any pair $(a,b)$ with $F(a,b)=G(a,b)=0$ is such that $Res(F,G,y)(a)=0$. For more information on resultants and elimination theory, the reader can refer to \cite{CLO2007}.


\section{New classes of $0$-APN power functions over $\mathbb{F}_{2^n}$}\label{third}

In this section, based on the multivariate method \cite{Dobbertin2002} and resultant elimination, we give some new classes of $0$-APN power functions over $\mathbb{F}_{2^n}$.

\subsection{The case of $n=4k$}
In this subsection, a new class of $0$-APN power functions over $\mathbb{F}_{2^n}$ is presented  for $n=4k$.

\begin{theorem}\label{th1.1}
Let $k$ be an odd integer with $n=4k$. Then
\begin{eqnarray*}
f(x)= x^{2^{2k-1}-2^{k-1}-1}
\end{eqnarray*}
is a $0$-APN function over $\mathbb{F}_{2^n}$.
\end{theorem}

\begin{proof}It suffices to prove that the equation
\begin{eqnarray}\label{eq:1.0}
(x + 1)^{2^{2k-1}-2^{k-1}-1} + x^{2^{2k-1}-2^{k-1}-1} + 1 = 0
\end{eqnarray}
has no solution in $\mathbb{F}_{2^n}\backslash\{0,1\}$. Assume that $x\neq0,1$ is a solution of Eq. \eqref{eq:1.0}. Thus Eq. \eqref{eq:1.0} can be reduced to
\begin{eqnarray}\label{eq:1.1}
x^{2^{2k-1}+2^{k-1}} + x^{2^{2k-1}+1} + x^{2^{2k-1}} + x^{2^{k}+2}  + x^{2^{k}+1} + x^{2^{k-1}+2} = 0.
\end{eqnarray}
Raising the square to Eq. \eqref{eq:1.1} leads to
\begin{eqnarray}\label{eq:1.2}
x^{2^{2k}+2^{k}} + x^{2^{2k}+2} + x^{2^{2k}} + x^{2^{k+1}+4}  + x^{2^{k+1}+2} + x^{2^{k}+4} = 0.
\end{eqnarray}
Let $y=x^{2^k}$, $z=y^{2^k}$ and $u=z^{2^k}$. Then $u^{2^k}=x$, and raising the $2^k$-th, $2^{2k}$-th and $2^{3k}$-th power to Eq. \eqref{eq:1.2} respectively gives
\begin{subequations}  \label{eq:1.3}
\begin{empheq}[left=\empheqlbrace]{align}
yz + x^2z + z + x^4y^2 + x^2y^2 + x^4y = 0,  \label{eq:1.3a} \\
zu + y^2u + u + y^4z^2 + y^2z^2 + y^4z = 0,  \label{eq:1.3b} \\
ux + z^2x + x + z^4u^2 + z^2u^2 + z^4u = 0,  \label{eq:1.3c} \\
xy + u^2y + y + u^4x^2 + u^2x^2 + u^4x = 0.  \label{eq:1.3d}
\end{empheq}
\end{subequations}
Computing the resultants of Eq. \eqref{eq:1.3b} and Eq. \eqref{eq:1.3c}, Eq. \eqref{eq:1.3b} and Eq. \eqref{eq:1.3d} with respect to $u$ respectively, we obtain
\begin{eqnarray*}  
Res_{1}(x,y,z)&=& xy^6z^2 + xy^6z + xy^4z^3 + xy^4z + xy^4 + xy^2z^3 + xy^2z^2 + xz^4 + x + y^8z^8 \nonumber \\
&& + y^8z^4 + y^6z^6 + y^6z^5 + y^4z^8 + y^4z^7 + y^4z^5 + y^2z^7 + y^2z^6 = 0,  \\   
Res_{2}(x,y,z)&=& y(y+1)(x^2y^{14}z^8 + x^2y^{14}z^4 + x^2y^{13}z^8 + x^2y^{13}z^4 + x^2y^{12}z^8 + x^2y^{12}z^4 \nonumber \\
&&+ x^2y^{11}z^8 + x^2y^{11}z^4 + x^2y^{10}z^8 + x^2y^{10}z^2 + x^2y^9z^8 + x^2y^9z^2 + x^2y^8z^8 \nonumber \\
&&+ x^2y^8z^2 + x^2y^7z^8 + x^2y^7z^2 + x^2y^6z^6 + x^2y^6z^4 + x^2y^5z^6 + x^2y^5z^4 + x^2y^4z^6 \nonumber \\
&&+ x^2y^4z^4 + x^2y^3z^6 + x^2y^3z^4 + xy^{14}z^8 +xy^{14}z^4 + xy^{13}z^8 + xy^{13}z^4 + xy^{12}z^8 \nonumber \\
&&+ xy^{12}z^4 + xy^{11}z^8 + xy^{11}z^4 + xy^{10}z^8 + xy^{10}z^4 + xy^9z^8 + xy^9z^4 + xy^8z^8 \nonumber \\
&&+ xy^8z^4 + xy^7z^8 + xy^7z^4 + xy^7 + xy^6z^4 + xy^6 + xy^5z^4 +xy^5 + xy^4z^4 + xy^4 \nonumber \\
&&+ xy^3z^4 + xy^3 + xy^2z^4 + xy^2 + xyz^4+ xy + xz^4 + x + y^{11}z^4 + y^{11}z^2 + y^{10}z^4 \nonumber \\
&&+ y^{10}z^2 + y^9z^4 + y^9z^2 + y^8z^4 + y^8z^2 + y^7z^6 + y^7 + y^6z^6 + y^6 + y^5z^6 + y^5 \nonumber \\
&&+ y^4z^6 + y^4 + y^3z^4 + y^3 + y^2z^4 + y^2 + yz^4 + y + z^4 + 1) = 0. 
\end{eqnarray*}
Note that $x,y\not\in\mathbb{F}_{2}$. Thus we compute the resultants Eq. \eqref{eq:1.3a} and $Res_{1}(x,y,z)$, Eq. \eqref{eq:1.3a} and $Res_{2}(x,y,z)/y(y+1)$ with respect to $z$, and obtain
\begin{subequations}  \label{eq:1.5}
\begin{empheq}[left=\empheqlbrace]{align}
Res_{1}(x,y) = 0, \\
Res_{2}(x,y) = 0,   
\end{empheq}
\end{subequations}
where the resultants $Res_{1}(x,y)$ and $Res_{2}(x,y)$ are listed in Appendix.
If $x\in \mathbb{F}_{2^2}$, then $x^{2^{2k}}=x$ and $x^{2^k}=x^2$ for $k$ being odd. Thereby we derive from Eq. \eqref{eq:1.2} that $x^2+x=0$, which means $x\in \mathbb{F}_{2}$, a contradiction. Thus $x\not\in \mathbb{F}_{2^2}$.

Assume that $xy + y + 1=0$. Then taking it to the $2^k$-th power derives $yz + z + 1=0$.   Plugging these two equations into Eq. \eqref{eq:1.3a}, we deduce $x^6 + x^5 + x^4 + x^3 + x^2 + x=x(x+1)(x^2 + x + 1)^2=0$. Notice that $x^2 + x + 1$ is an irreducible polynomial on $\mathbb{F}_{2}$, this yields $x\in \mathbb{F}_{2^2}$, which leads to a contradiction.  Hence $xy + y + 1\neq0$. Similarly, we can prove $xy + x + 1\neq0$.

Computing the resultant $Res_{1}(x,y)/(xy + y + 1)(xy + x + 1)$ and $Res_{2}(x,y)/(xy + y + 1)(xy + x + 1)$ with respect to $y$, and then the resultant can be decomposed into the product of some irreducible factors in $\mathbb{F}_{2}$ as
\begin{eqnarray}\label{eq:1.6}
&&  x^{352}(x + 1)^{352}(x^2 + x + 1)^{162}(x^6 + x + 1)^2(x^6 + x^3 + 1)^2(x^6 + x^4 + x^2 + x + 1)^2  \nonumber \\
&& \cdot(x^6 + x^4 + x^3 + x + 1)^2(x^6 + x^5 + 1)^2(x^6 + x^5 + x^2 + x + 1)^2(x^6 + x^5 + x^3 + x^2 + 1)^2  \nonumber \\
&& \cdot(x^6 + x^5 + x^4 + x + 1)^2 (x^6 + x^5 + x^4 + x^2 + 1)^2 =0
\end{eqnarray}
by MAGMA computation. Recall that $x\not\in \mathbb{F}_{2^2}$, we have $x^2 + x + 1\neq0$.
Suppose one of the equations $x^6 + x + 1=0$, $x^6 + x^3 + 1=0$, $x^6 + x^4 + x^2 + x + 1=0$, $x^6 + x^4 + x^3 + x + 1=0$, $x^6 + x^5 + 1=0$, $x^6 + x^5 + x^2 + x + 1=0$, $x^6 + x^5 + x^3 + x^2 + 1=0$, $x^6 + x^5 + x^4 + x + 1=0$ and $x^6 + x^5 + x^4 + x^2 + 1=0$ holds.  Then $x\in \mathbb{F}_{2^6}$.

When $k\equiv 0~({\rm mod}\ 3)$, the solutions of Eq. \eqref{eq:1.6} belong into $\mathbb{F}_{2^6}$. Then $x^{2^{2k}}=x$ and  $x^{2^k}=x^8$ for $k$ being odd. It follows from \eqref{eq:1.2} that
\begin{eqnarray*}
x^{9} + x^{3} + x + x^{20} + x^{18} + x^{12} = 0.
\end{eqnarray*}
The equation can be decomposed into the following product of irreducible factors on $\mathbb{F}_{2}$ as
\begin{eqnarray}\label{eq:1.7}
x(x + 1)(x^9 + x + 1)(x^9 + x^8 + 1) = 0.
\end{eqnarray}
The solutions of Eq. \eqref{eq:1.7} are in $\mathbb{F}_{2^9}$. Notice that $\mathbb{F}_{2^6}\cap\mathbb{F}_{2^9}=\mathbb{F}_{2^3}$.
For $x\in \mathbb{F}_{2^3}$, we have $x^{2^{2k}}=x$ and  $x^{2^k}=x$. Hence we deduce from \eqref{eq:1.2} that
\begin{eqnarray*}
x^{2} + x^{3} + x + x^{6} + x^{4} + x^{5} = x(x+1)(x^2 + x + 1)^2=0.
\end{eqnarray*}
It can be checked that $x^2 + x + 1$ is irreducible on $\mathbb{F}_{2}$. Then we obtain $x\in \mathbb{F}_{2^2}$, a contradiction.

When $k\not\equiv 0~({\rm mod}\ 3)$, the solutions of Eq. \eqref{eq:1.6} belong into $\mathbb{F}_{2^6}\cap\mathbb{F}_{2^n}=\mathbb{F}_{2^2}$.
This is a contradiction. Therefore, Eq. \eqref{eq:1.0} has no solution in $\mathbb{F}_{2^n}\backslash\{0,1\}$.
This completes the proof.
\end{proof}

\subsection{The case of $n=2k+1$}
In this subsection, four new classes of $0$-APN power functions over $\mathbb{F}_{2^n}$ are given for $n=2k+1$.

\begin{theorem}\label{th:2.1}
Let $n, k$ be positive integers with $n=2k+1$. Then
\begin{eqnarray*}
f(x)= x^{2^{2k-1}+2^{k}+1}
\end{eqnarray*}
is a $0$-APN function over $\mathbb{F}_{2^n}$.
\end{theorem}

\begin{proof}
We need to prove that the equation
\begin{eqnarray}\label{1.0}
f(x + 1) + f(x) + 1 = (x + 1)^{2^{2k-1}+2^{k}+1} + x^{2^{2k-1}+2^{k}+1} + 1 = 0
\end{eqnarray}
has no solution in $\mathbb{F}_{2^n}\backslash\{0,1\}$. Assume that $x$ is a solution of Eq. \eqref{1.0} with $x\neq0,1$. Eq. \eqref{1.0} becomes
\begin{eqnarray}\label{1.1}
x^{2^{2k-1}+2^{k}} + x^{2^{2k-1}+1} + x^{2^{k}+1} + x^{2^{2k-1}} + x^{2^{k}} + x = 0.
\end{eqnarray}
Raising the fourth power to Eq. \eqref{1.1} gives
\begin{eqnarray}\label{1.2}
x^{1+2^{k+2}} + x^{5} + x^{2^{k+2}+4} + x + x^{2^{k+2}} + x^4 = 0.
\end{eqnarray}
Let $y=x^{2^{k+1}}$. Then $y^{2^{k+1}}=x^2$.  Eq. \eqref{1.2} can be written as
\begin{eqnarray}\label{1.3}
xy^2 + x^{5} + x^{4}y^2 + x + y^2 + x^4 = 0.
\end{eqnarray}
Raising the $2^{k+1}$-th power to Eq. \eqref{1.3} gives
\begin{eqnarray}\label{1.4}
yx^4 + y^{5} + y^{4}x^4 + y + x^4 + y^4 = 0.
\end{eqnarray}
Computing the resultant of Eq. \eqref{1.3} and Eq. \eqref{1.4} with respect to $y$, by MAGMA computation, and then the resultant can be decomposed into the product of some irreducible factors in $\mathbb{F}_{2}$ as
\begin{eqnarray}\label{1.5}
&& x(x + 1)(x^5 + x^2 + 1)(x^5 + x^3 + 1)(x^5 + x^3 + x^2 + x + 1)(x^5 + x^4 + x^2 + x + 1) \nonumber \\
&& (x^5 + x^4 + x^3 + x + 1)(x^5 + x^4 + x^3 + x^2 + 1)=0.
\end{eqnarray}
Observe that $x\neq0,1$, we have $x^5 + x^2 + 1=0$, $x^5 + x^3 + 1=0$, $x^5 + x^3 + x^2 + x + 1=0$, $x^5 + x^4 + x^2 + x + 1=0$, $x^5 + x^4 + x^3 + x + 1=0$ or $x^5 + x^4 + x^3 + x^2 + 1=0$.
Suppose that one of the five equations holds. Then $x\in \mathbb{F}_{2^5}$.

When $k\not\equiv 2~({\rm mod}\ 5)$, we get $5\nmid n$. It leads to $\mathbb{F}_{2^5}\cap\mathbb{F}_{2^n}=\mathbb{F}_{2}$, which contradicts  with $x\neq0, 1$.  Therefore, Eq. \eqref{1.0} has no solution in $\mathbb{F}_{2^n}\backslash\{0,1\}$.

When $k\equiv 2~({\rm mod}\ 5)$, the solutions of Eq. \eqref{1.5} belong into $\mathbb{F}_{2^5}$. Note that $k+3\equiv 0~({\rm mod}\ 5)$. Raising the square to Eq. \eqref{1.2} derives
\begin{eqnarray*}
x^{2+2^{k+3}} + x^{10} + x^{2^{k+3}+8} + x^2 + x^{2^{k+3}} + x^8 = 0.
\end{eqnarray*}
Since $x\in\mathbb{F}_{2^5}$, the equation can be written as
\begin{eqnarray*}
x^{3} + x^{10} + x^{9} + x^2 + x + x^8 = 0.
\end{eqnarray*}
The equation can be decomposed into the product of irreducible factors in $\mathbb{F}_{2}$ as
\begin{eqnarray}\label{1.7}
x(x + 1)(x^2 + x + 1)(x^3 + x + 1)(x^3 + x^2 + 1) = 0.
\end{eqnarray}
The solutions of Eq. \eqref{1.7} are in $\mathbb{F}_{2^2}$ or $\mathbb{F}_{2^3}$. Notice that $\mathbb{F}_{2^2}\cap\mathbb{F}_{2^5}=\mathbb{F}_{2}$ and  $\mathbb{F}_{2^3}\cap\mathbb{F}_{2^5}=\mathbb{F}_{2}$. Hence we obtain that $x=0, 1$ are the solutions of Eq. \eqref{1.0}, which contradicts with $x\neq0, 1$.
Therefore, Eq. \eqref{1.0} has no solution in $\mathbb{F}_{2^n}\backslash\{0,1\}$.
This completes the proof.
\end{proof}


\begin{theorem}\label{th:2.2}
Let $n$ and $k$ be positive integers with $k\not\equiv 1~({\rm mod}\ 3)$ and $n=2k+1$. Then
\begin{eqnarray*}
f(x)= x^{2^{2k}+2^{k+1}+1}
\end{eqnarray*}
is a $0$-APN function over $\mathbb{F}_{2^n}$.
\end{theorem}

\begin{proof}
It suffices to show that the equation
\begin{eqnarray}\label{2.0}
f(x + 1) + f(x) + 1 = (x + 1)^{2^{2k}+2^{k+1}+1} + x^{2^{2k}+2^{k+1}+1} + 1 = 0
\end{eqnarray}
has no solution in $\mathbb{F}_{2^n}\backslash\{0,1\}$. Eq. \eqref{2.0}  can be simplified as
\begin{eqnarray}\label{2.1}
x^{2^{2k}+2^{k+1}} + x^{2^{2k}+1} + x^{2^{k+1}+1} + x^{2^{2k}} + x^{2^{k+1}} + x = 0.
\end{eqnarray}
Raising the square to  Eq. \eqref{2.1} results in
\begin{eqnarray}\label{2.2}
x^{1+2^{k+2}} + x^{3} + x^{2^{k+2}+2} + x + x^{2^{k+2}} + x^2 = 0.
\end{eqnarray}
Let $y=x^{2^{k+1}}$. Then $y^{2^{k+1}}=x^2$. Eq. \eqref{2.2} can be written as
\begin{eqnarray}\label{2.3}
xy^2 + x^{3} + x^{2}y^2 + x + y^2 + x^2 = 0.
\end{eqnarray}
Raising the $2^{k+1}$-th power to Eq. \eqref{2.3} gives
\begin{eqnarray}\label{2.4}
yx^2 + y^{3} + y^{2}x^2 + y + x^2 + y^2 = 0.
\end{eqnarray}
Computing the resultant of Eq. \eqref{2.3} and Eq. \eqref{2.4} with respect to $y$, and then with the help of MAGMA, the resultant can be decomposed into the following product of irreducible factors in $\mathbb{F}_{2}$ as
\begin{eqnarray}\label{2.5}
 x(x + 1)(x^2 + x + 1)^4(x^3 + x + 1)(x^3 + x^2 + 1)=0.
\end{eqnarray}
Observe that $x\neq0,1$, we obtain $x^2 + x + 1=0$, $x^3 + x + 1=0$ or $x^3 + x^2 + 1=0$.
Thus the solutions of Eq. \eqref{2.5} are in $\mathbb{F}_{2^2}$ or $\mathbb{F}_{2^3}$. We remark that $\mathbb{F}_{2^2}\cap\mathbb{F}_{2^n}=\mathbb{F}_{2}$ for $n$ being odd, and $\mathbb{F}_{2^3}\cap\mathbb{F}_{2^n}=\mathbb{F}_{2}$ for $n=2k+1$ and $k\not\equiv 1~({\rm mod}\ 3)$. Hence, Eq. \eqref{2.5} has no solution when $x\neq0, 1$.
It follows that Eq. \eqref{2.0} has no solution in $\mathbb{F}_{2^n}\backslash\{0,1\}$.
We complete the proof.
\end{proof}


\begin{theorem}\label{th:2.3}
Let $n$ and $k$ be positive integers with $k\not\equiv 1~({\rm mod}\ 3)$ and $n=2k+1$. Then
\begin{eqnarray*}
f(x)= x^{2^{k+1}-2^{k-1}-1}
\end{eqnarray*}
is a $0$-APN function over $\mathbb{F}_{2^n}$.
\end{theorem}

\begin{proof}
We need to show that the equation
\begin{eqnarray}\label{3.0}
(x + 1)^{2^{k+1}-2^{k-1}-1} + x^{2^{k+1}-2^{k-1}-1} + 1 = 0
\end{eqnarray}
has no solution in $\mathbb{F}_{2^n}\backslash\{0,1\}$. Suppose that $x\neq0,1$ is a solution of Eq. \eqref{3.0}. Then Eq. \eqref{3.0} can be simplified as
\begin{eqnarray}\label{3.1}
x^{2^{k+1}+2^{k-1}} + x^{2^{k+1}+1} + x^{2^{k+1}} + x^{2^{k}+2} + x^{2^{k}+1} + x^{2^{k-1}+2} = 0.
\end{eqnarray}
Raising the fourth power to Eq. \eqref{3.1} gives
\begin{eqnarray}\label{3.2}
x^{5\cdot2^{k+1}} + x^{4\cdot2^{k+1}+4} + x^{4\cdot2^{k+1}} + x^{2\cdot2^{k+1}+8} + x^{2\cdot2^{k+1}+4} + x^{2^{k+1}+8} = 0.
\end{eqnarray}
Let $y=x^{2^{k+1}}$. Then $y^{2^{k}}=x$.  Eq. \eqref{3.2} can be written as
\begin{eqnarray}\label{3.3}
y^5 + x^4y^{4} + y^4 + x^8y^2 + x^4y^2 + x^8y  = 0.
\end{eqnarray}
Raising the $2^{k}$-th power to Eq. \eqref{3.3} results in
\begin{eqnarray}\label{3.4}
x^5 + y^2x^{4} + x^4 + y^4x^2 + x^2y^2 + xy^4 = 0.
\end{eqnarray}
Computing the resultant of Eq. \eqref{3.3} and Eq. \eqref{3.4} with respect to $y$, and then by MAGMA,  the resultant can be decomposed into the product of irreducible factors in $\mathbb{F}_{2}$ as
\begin{eqnarray*}
 x^{17}(x + 1)^{17}(x^2 + x + 1)^4(x^3 + x + 1)(x^3 + x^2 + 1)=0.
\end{eqnarray*}
Note that $x\not\in\mathbb{F}_{2}$, we have $x^2 + x + 1=0$, $x^3 + x + 1=0$ or $x^3 + x^2 + 1=0$.

Suppose $x^2 + x + 1=0$. Then $x \in\mathbb{F}_{2^2} \cap\mathbb{F}_{2^n}=\mathbb{F}_{2}$ for $n$ being odd, which contradicts with $x\not\in\mathbb{F}_{2}$.

Suppose  $x^3 + x + 1=0$ or $x^3 + x^2 + 1=0$. Then $x\in \mathbb{F}_{2^3}$. Since $n=2k+1$ and $k\not\equiv 1~({\rm mod}\ 3)$, it can be verified that $x\in\mathbb{F}_{2^3} \cap\mathbb{F}_{2^n}=\mathbb{F}_{2}$, which leads to a contradiction. Therefore Eq. \eqref{3.0} has no solution in $\mathbb{F}_{2^n}\backslash\{0,1\}$.
The proof is completed.
\end{proof}


\begin{theorem}\label{th:2.4}
Let $n$ and $k$ be positive integers with $k\not\equiv 4~({\rm mod}\ 9)$ and $n=2k+1$. Then
\begin{eqnarray*}
f(x)= x^{2^{2k}-2^{k+1}-1}
\end{eqnarray*}
is a $0$-APN function over $\mathbb{F}_{2^n}$.
\end{theorem}

\begin{proof}
It suffices to show that the equation
\begin{eqnarray}\label{4.0}
(x + 1)^{2^{2k}-2^{k+1}-1} + x^{2^{2k}-2^{k+1}-1} + 1 = 0
\end{eqnarray}
has no solution in $\mathbb{F}_{2^n}\backslash\{0,1\}$. Assume that $x\neq0,1$ is a solution of Eq. \eqref{4.0}. Eq. \eqref{4.0}  can be simplified as
\begin{eqnarray}\label{4.1}
x^{2^{2k}+2^{k+1}} + x^{2^{2k}+1} + x^{2k} + x^{2^{k+2}+2} + x^{2^{k+2}+1} + x^{2^{k+1}+2} = 0.
\end{eqnarray}
Raising the square to Eq. \eqref{4.1} leads to
\begin{eqnarray}\label{4.2}
x^{2^{k+2}+1} + x^{3} + x  + x^{4\cdot2^{k+1}+4} + x^{4\cdot2^{k+1}+2} + x^{2\cdot2^{k+1}+4} = 0.
\end{eqnarray}
Let $y=x^{2^{k+1}}$. Then $y^{2^{k+1}}=x^2$, and Eq. \eqref{4.2} can be written as
\begin{eqnarray}\label{4.3}
xy^{2} + x^3 + x + x^4y^4 + x^2y^4 +x^4y^2  = 0.
\end{eqnarray}
Raising the $2^{k+1}$-th power to Eq. \eqref{4.3} gives
\begin{eqnarray}\label{4.4}
yx^4 + y^3 + y + y^4x^8 + y^2x^8 + y^4x^4 = 0.
\end{eqnarray}
Computing the resultant of Eq. \eqref{4.3} and Eq. \eqref{4.4} with respect to $y$, and then the resultant can be decomposed into the product of irreducible factors in $\mathbb{F}_{2}$ as
\begin{eqnarray}\label{4.5}
 x^{7}(x + 1)^{7}(x^3 + x + 1)^2(x^3 + x^2 + 1)^2(x^9 + x + 1)(x^9 + x^8 + 1)=0
\end{eqnarray}
by MAGMA. Note that $x\not\in\mathbb{F}_{2}$, we have $x^3 + x + 1=0$, $x^3 + x^2 + 1=0$,  $x^9 + x + 1=0$ or $x^9 + x^8 + 1=0$.
If $x^3 + x + 1=0$ or $x^3 + x^2 + 1=0$, then $x\in \mathbb{F}_{2^3}$. If $x^9 + x + 1=0$ or $x^9 + x^8 + 1=0$, then $x\in \mathbb{F}_{2^9}$.

When $k\not\equiv 1~({\rm mod}\ 3)$, we have $\gcd(3,n)=1$ and $\mathbb{F}_{2^3} \cap\mathbb{F}_{2^n}=\mathbb{F}_{2^9} \cap\mathbb{F}_{2^n}=\mathbb{F}_{2}$.
Therefore Eq. \eqref{4.0} has no solution in $\mathbb{F}_{2^n}\backslash\{0,1\}$.

When $k\equiv 1~({\rm mod}\ 3)$ and $k\not\equiv 4~({\rm mod}\ 9)$, we have $\gcd(n, 9)=3$. This means that the solutions of Eq. \eqref{4.5} belong into $\mathbb{F}_{2^3}$. Note that $k+2\equiv 0~({\rm mod}\ 3)$. We derive from  Eq. \eqref{4.1} that
$$ x^2 + x^3 + x + x^6 + x^5 + x^4 =0, $$
which can be simplified as
$$ x(x + 1)(x^2 + x + 1)^2 =0. $$
Observe that $x^2 + x + 1$ is irreducible in $\mathbb{F}_{2}$. The solutions of the above equation lie in $\mathbb{F}_{2^2}$. However, $\mathbb{F}_{2^2} \cap\mathbb{F}_{2^3}=\mathbb{F}_{2}$.
Hence, Eq. \eqref{4.0} has no solution in $\mathbb{F}_{2^n}\backslash\{0,1\}$. The proof is completed.
\end{proof}


\subsection{The case of $n=3k-1$}

In this subsection, we present four new classes of $0$-APN functions over $\mathbb{F}_{2^n}$ of $n=3k-1$.

\begin{theorem}\label{th:2.5}
Let $n$ and $k$ be positive integers with $n=3k-1$. Then
\begin{eqnarray*}
f(x)= x^{2^{2k}+2^{k+1}+1}
\end{eqnarray*}
is a $0$-APN function over $\mathbb{F}_{2^n}$.
\end{theorem}

\begin{proof}
We need to prove that the equation
\begin{eqnarray}\label{5.0}
(x + 1)^{2^{2k}+2^{k+1}+1} + x^{2^{2k}+2^{k+1}+1} + 1 = 0
\end{eqnarray}
has no solution in $\mathbb{F}_{2^n}\backslash\{0,1\}$. Eq. \eqref{5.0}  can be written as
\begin{eqnarray}\label{5.1}
x^{2^{2k}+2^{k+1}} + x^{2^{2k}+1} + x^{2^{k+1}+1} + x^{2^{2k}} + x^{2^{k+1}} + x = 0.
\end{eqnarray}
Let $y=x^{2^k}$ and $z=y^{2^k}$. Then $z^{2^k}=x^2$, and Eq. \eqref{5.1} turns into
\begin{eqnarray}\label{5.3}
y^2z + xz + xy^2 + z + y^2 + x = 0.
\end{eqnarray}
Raising the $2^{2k}$-th power to  Eq. \eqref{5.3} leads to
\begin{eqnarray}\label{5.5}
x^4y^2 + zy^{2} + zx^4 + y^2 + x^4 + z = 0.
\end{eqnarray}
Computing the resultant of Eq. \eqref{5.3} and Eq. \eqref{5.5} with respect to $z$, and then decomposing the resultant into
\begin{eqnarray*}
 x(x + 1)(x^2 + x + 1)(y^2 + y + 1)^2=0
\end{eqnarray*}
with the help of MAGMA. For $x\not\in\mathbb{F}_{2}$, we have $x^2 + x + 1=0$ or $y^2 + y + 1=0$. Obviously, the polynomial $x^2 + x + 1 $ is irreducible on $\mathbb{F}_{2}$.

Suppose that $x^2 + x + 1=0$. If $k$ is even, then $x \in\mathbb{F}_{2^2} \cap\mathbb{F}_{2^n}=\mathbb{F}_{2}$, a contradiction.
If $k$ is odd, then $x \in\mathbb{F}_{2^2}\backslash\{0,1\}$. Let $\omega\in\mathbb{F}_{2^2}\backslash\{0,1\}$ and $\omega^3=1$. Plugging $x=\omega$ into Eq. \eqref{5.0} derives $1=0$ for $k$ being odd, which is impossible. Similarly, we can prove $y^2 + y + 1\neq0$. Therefore, Eq. \eqref{5.0} has no solution in $\mathbb{F}_{2^n}\backslash\{0,1\}$.
The proof is completed.
\end{proof}

\begin{theorem}\label{th:2.6}
Let $n$ be an integer and $k$ be even with $n=3k-1$. Then
\begin{eqnarray*}
f(x)= x^{2^{2k+1}+2^{k+1}+1}
\end{eqnarray*}
is a $0$-APN function over $\mathbb{F}_{2^n}$.
\end{theorem}

\begin{proof}
We will verify that the equation
\begin{eqnarray}\label{6.0}
(x + 1)^{2^{2k+1}+2^{k+1}+1} + x^{2^{2k+1}+2^{k+1}+1} + 1 = 0
\end{eqnarray}
has no solution in $\mathbb{F}_{2^n}\backslash\{0,1\}$. Eq. \eqref{6.0}  can become
\begin{eqnarray}\label{6.1}
x^{2^{2k+1}+2^{k+1}} + x^{2^{2k+1}+1} + x^{2^{k+1}+1} + x^{2^{2k+1}} + x^{2^{k+1}} + x = 0.
\end{eqnarray}
Let $y=x^{2^k}$ and $z=y^{2^k}$. Then $z^{2^k}=x^2$, and Eq. \eqref{6.1} can be written as
\begin{eqnarray}\label{6.2}
y^2z^2 + xz^2 + xy^2 + z^2 + y^2 + x = 0.
\end{eqnarray}
Raising the $2^{k}$-th power to Eq. \eqref{6.2} gives
\begin{eqnarray}\label{6.3}
z^2x^4 + yx^4 + yz^2 + x^4 + z^2 + y = 0.
\end{eqnarray}
Raising the $2^{k}$-th power to Eq. \eqref{6.3} obtains
\begin{eqnarray}\label{6.4}
x^4y^4 + zy^4 + zx^4 + y^4 + x^4 + z = 0.
\end{eqnarray}
Computing the resultants of Eq. \eqref{6.2} and Eq. \eqref{6.3}, Eq. \eqref{6.2} and Eq. \eqref{6.4} with respect to $z$ respectively, we obtain
\begin{subequations}  \label{eq:6.5}
\begin{empheq}[left=\empheqlbrace]{align}
& (x + y + 1)^2(xy + 1)^2(xy + x + y)^2(x^2 + x + 1)^2 = 0,  \label{6.5a} \\
& (y^2 + y + 1)^2(x^3 + x^2y^2 + xy^2 + x + 1)(x^3y^2 + x^2y^2 + x^2 + x + y^2)\notag \\
& \cdot(x^3y^2 + x^3 + x^2 + xy^2 + y^2 + 1) = 0.  \label{6.5b}
\end{empheq}
\end{subequations}
Computing the resultant of Eq. \eqref{6.5a} and Eq. \eqref{6.5b} with respect to $y$, and then decomposing it into
\begin{eqnarray*}
 x^6(x + 1)^6(x^2 + x + 1)^{68}=0.
\end{eqnarray*}
It is clear that the polynomial $x^2 + x + 1$ is irreducible on $\mathbb{F}_{2}$. Suppose that $x^2 + x + 1=0$. Then $x \in\mathbb{F}_{2^2}$. However, $\mathbb{F}_{2^2}\cap\mathbb{F}_{2^n}=\mathbb{F}_{2}$ for $k$ being even, which is a contradiction. Therefore, $x^2 + x + 1\neq0$. Thereby Eq. \eqref{6.0} has no solution in $\mathbb{F}_{2^n}\backslash\{0,1\}$.
The proof is completed.
\end{proof}

\begin{theorem}\label{th:2.7}
Let $n$ be an integer and $k$ be even with $n=3k-1$. Then
\begin{eqnarray*}
f(x)= x^{2^{2k+1}+2^{k}+1}
\end{eqnarray*}
is a $0$-APN function over $\mathbb{F}_{2^n}$.
\end{theorem}

\begin{proof}
It suffices to prove that the equation
\begin{eqnarray}\label{7.0}
(x + 1)^{2^{2k+1}+2^{k}+1} + x^{2^{2k+1}+2^{k}+1} + 1 = 0
\end{eqnarray}
has no solution in $\mathbb{F}_{2^n}\backslash\{0,1\}$. Eq. \eqref{7.0}  can be written as
\begin{eqnarray}\label{7.1}
x^{2^{2k+1}+2^{k}} + x^{2^{2k+1}+1} + x^{2^{k}+1} + x^{2^{2k+1}} + x^{2^{k}} + x = 0.
\end{eqnarray}
Let $y=x^{2^k}$ and $z=y^{2^k}$. Then $z^{2^k}=x^2$, and taking the $2^k$-th and $2^{2k}$-th power to Eq. \eqref{7.1} respectively, we deduce
\begin{subequations}  \label{eq:7.2}
\begin{empheq}[left=\empheqlbrace]{align}
yz^2 + xz^2 + xy + z^2 + y + x = 0,  \label{7.2a} \\
zx^4 + yx^4 + yz + x^4 + z + y = 0,  \label{7.2b} \\
x^2y^4 + zy^4 + zx^2 + y^4 + x^2 + z = 0.  \label{7.2c}
\end{empheq}
\end{subequations}
Computing the resultants of Eq. \eqref{7.2a} and Eq. \eqref{7.2b}, Eq. \eqref{7.2a} and Eq. \eqref{7.2c} with respect to $z$ respectively, we obtain
\begin{subequations}  \label{eq:7.3}
\begin{empheq}[left=\empheqlbrace]{align}
 x^9y^2 + x^9y + x^8y^3 + x^8y^2 + x^8 + xy^3 + xy + x + y^2 + y = 0,  \label{7.3a} \\
 x^5y^8 + x^5y + x^4y^9 + x^4y^8 + x^4 + xy^9 + xy + x + y^8 + y = 0.  \label{7.3b}
\end{empheq}
\end{subequations}
Computing the resultant of Eq. \eqref{7.3a} and Eq. \eqref{7.3b} with respect to $y$, and then by MAGMA computation, the resultant can be decomposed into the product of irreducible factors in $\mathbb{F}_{2}$ as
\begin{eqnarray}\label{7.4}
 x^6(x + 1)^6(x^2 + x + 1)^{3}(x^4 + x + 1)^6(x^4 + x^3 + 1)^6(x^4 + x^3 + x^2 + x + 1)^6=0.
\end{eqnarray}
Observe that $x \not\in\mathbb{F}_{2}$, we have $x^2 + x + 1=0$, $x^4 + x + 1=0$, $x^4 + x^3 + 1=0$ or $x^4 + x^3 + x^2 + x + 1=0$.
Hence the solutions of Eq. \eqref{7.4} lie in $\mathbb{F}_{2^2}$ or $\mathbb{F}_{2^4}$. But $\mathbb{F}_{2^2} \cap\mathbb{F}_{2^n}=\mathbb{F}_{2}$ and $\mathbb{F}_{2^4} \cap\mathbb{F}_{2^n}=\mathbb{F}_{2}$ for $k$ being even. It leads to $x\in\mathbb{F}_{2}$, which is a contradiction. Therefore, Eq. \eqref{7.0} has no solution in $\mathbb{F}_{2^n}\backslash\{0,1\}$. This completes the proof.
\end{proof}

\begin{theorem}\label{th:2.16}
Let $n$ be an integer and $k$ be even with $n=3k-1$. Then
\begin{eqnarray*}
f(x)= x^{3\cdot2^{2k}+1}
\end{eqnarray*}
is a $0$-APN function over $\mathbb{F}_{2^n}$.
\end{theorem}

\begin{proof}
We will show that the equation
\begin{eqnarray}\label{16.0}
(x + 1)^{3\cdot2^{2k}+1} + x^{3\cdot2^{2k}+1} + 1 = 0
\end{eqnarray}
has no solution in $\mathbb{F}_{2^n}\backslash\{0,1\}$. We deduce from Eq. \eqref{16.0}  that
\begin{eqnarray}\label{16.1}
x^{^{3\cdot2^{2k}}} + x^{2^{2k+1}+1} + x^{2^{2k}+1} + x^{2^{2k+1}} + x^{2^{2k}} + x = 0.
\end{eqnarray}
Let $y=x^{2^k}$ and $z=y^{2^k}$. Then $z^{2^k}=x^2$, and taking the $2^k$-th and $2^{2k}$-th power to Eq. \eqref{16.1} respectively obtains
\begin{subequations}  \label{eq:16.2}
\begin{empheq}[left=\empheqlbrace]{align}
z^3 + xz^2 + xz + z^2 + z + x = 0,  \label{16.2a} \\
x^6 + yx^4 + yx^2 + x^4 + x^2 + y = 0,  \label{16.2b} \\
y^6 + zy^4 + zy^2 + y^4 + y^2 + z = 0.  \label{16.2c}
\end{empheq}
\end{subequations}
Computing the resultant of Eq. \eqref{16.2a} and Eq. \eqref{16.2c} with respect to $z$ respectively gives
\begin{eqnarray} \label{16.3}
 (y^2 + y + 1)^8(x + y^2) = 0.
\end{eqnarray}
Computing the resultant of Eq. \eqref{16.2b} and Eq. \eqref{16.3} with respect to $y$, and then decomposing it into 
\begin{eqnarray*}
 x(x + 1)(x^2 + x + 1)^{53} =0.
\end{eqnarray*}
Note that the polynomial $x^2 + x + 1$ is irreducible over $\mathbb{F}_{2}$.
Assume that $x^2 + x + 1=0$. Then $x \in\mathbb{F}_{2^2}\backslash\mathbb{F}_{2}$. However, $x \in\mathbb{F}_{2^2}\cap\mathbb{F}_{2^n}=\mathbb{F}_{2}$ for $k$ being even, which is a contradiction. Therefore, Eq. \eqref{16.0} has no solution in $\mathbb{F}_{2^n}\backslash\{0,1\}$.
This completes the proof.
\end{proof}

\subsection{The case of $n=3k$}
In this subsection, four new classes of $0$-APN power functions over $\mathbb{F}_{2^n}$ are given for $n=3k$.

\begin{theorem}\label{th:2.9}
Let $n$ and $k$ be positive integers with $k\not\equiv 0~({\rm mod}\ 3)$ and $n=3k$. Then
\begin{eqnarray*}
f(x)= x^{2^{2k-1}-2^{k}-1}
\end{eqnarray*}
is a $0$-APN function over $\mathbb{F}_{2^n}$.
\end{theorem}

\begin{proof}
It suffices to prove that the equation
\begin{eqnarray}\label{9.0}
(x + 1)^{2^{2k-1}-2^{k}-1} + x^{2^{2k-1}-2^{k}-1} + 1 = 0
\end{eqnarray}
has no solution in $\mathbb{F}_{2^n}\backslash\{0,1\}$. Assume that $x\neq0,1$ is a solution of Eq. \eqref{9.0}. Eq. \eqref{9.0} can be simplified as
\begin{eqnarray}\label{9.1}
x^{2^{2k-1}+2^{k}} + x^{2^{2k-1}+1} + x^{2^{2k-1}} + x^{2^{k+1}+2}  + x^{2^{k+1}+1} + x^{2^{k}+2} = 0.
\end{eqnarray}
Raising the fourth power to Eq. \eqref{9.1} leads to
\begin{eqnarray}\label{9.2}
x^{2^{2k}+2^{k+1}} + x^{2^{2k}+2} + x^{2^{2k}} + x^{2^{k+2}+4}  + x^{2^{k+2}+2} + x^{2^{k+1}+4} = 0.
\end{eqnarray}
Let $y=x^{2^k}$ and $z=y^{2^k}$. Then $z^{2^k}=x$, and raising the $2^k$-th and $2^{2k}$-th power to Eq. \eqref{9.2} respectively derives
\begin{subequations}  \label{eq:9.3}
\begin{empheq}[left=\empheqlbrace]{align}
y^2z + x^2z + z + x^4y^4 + x^2y^4 + x^4y^2 = 0,  \label{9.3a} \\
z^2x + y^2x + x + y^4z^4 + y^2z^4 + y^4z^2 = 0,  \label{9.3b} \\
x^2y + z^2y + y + z^4x^4 + z^2x^4 + z^4x^2 = 0.  \label{9.3c}
\end{empheq}
\end{subequations}
Computing the resultants of Eq. \eqref{9.3a} and Eq. \eqref{9.3b}, Eq. \eqref{9.3a} and Eq. \eqref{9.3c} with respect to $z$ respectively obtains
\begin{eqnarray*}
Res_{1}(x,y) &=& x(x+1)(xy^2 + y^2 + 1)(xy^2 + x + 1)(x^{12}y^{16} + x^{12}y^8 + x^{11}y^{12} + x^{11}y^{10} + x^{11}y^8 \nonumber \\
&&+ x^{11}y^6 + x^{10}y^{16} + x^{10}y^{14} + x^{10}y^{12} + x^{10}y^{10} + x^{10}y^8 + x^{10}y^4 + x^9y^{14} + x^9y^6 \nonumber \\
&&+ x^8y^{16} + x^8y^{12} + x^7y^{12} + x^7y^{10} + x^7y^6 + x^7y^2 + x^6y^{16} + x^6y^{14} + x^6y^{12} + x^6y^8 \nonumber \\
&&+ x^6y^4 + x^6y^2 + x^6 + x^5y^{14} + x^5y^{10} + x^5y^6 + x^5y^4 + x^4y^4 + x^4 + x^3y^{10} + x^3y^2 \nonumber \\
&&+ x^2y^{12} + x^2y^8 + x^2y^6 + x^2y^4 + x^2y^2 + x^2 + xy^{10} + xy^8 + xy^6 + xy^4 + y^8 + 1) = 0,   \\ 
Res_{2}(x,y) &=& y(y+1)(x^2y + y + 1)(x^2y + x^2 + 1)(x^{16}y^{12} + x^{16}y^{10} + x^{16}y^8 + x^{16}y^6 + x^{14}y^{10} \nonumber \\
&&+ x^{14}y^9  + x^{14}y^6 + x^{14}y^5 + x^{12}y^{11} + x^{12}y^{10} + x^{12}y^8 + x^{12}y^7 + x^{12}y^6 + x^{12}y^2 \nonumber \\
&&+ x^{10}y^{11} + x^{10}y^{10} + x^{10}y^7  + x^{10}y^5 + x^{10}y^3 + x^{10}y + x^8y^{12} + x^8y^{11} + x^8y^{10} + x^8y^6 \nonumber \\
&&+ x^8y^2 + x^8y + x^8 + x^6y^{11} + x^6y^9 + x^6y^7 + x^6y^5 + x^6y^2 + x^6y + x^4y^{10} + x^4y^6 \nonumber \\
&&+ x^4y^5 + x^4y^4 + x^4y^2 + x^4y + x^2y^7 + x^2y^6 + x^2y^3 + x^2y^2 + y^6 + y^4 + y^2 + 1) = 0.  
\end{eqnarray*}
Recall that $x,y\not\in\mathbb{F}_{2}$. Thus we compute the resultant of  $Res_{1}(x,y)/x(x+1)$ and $Res_{2}(x,y)/y(y+1)$ with respect to $y$, by MAGMA computation, and then decomposing it into the product of irreducible factors in $\mathbb{F}_{2}$ as
\begin{eqnarray}\label{9.5}
&& x^{122}(x + 1)^{122}(x^2 + x + 1)^{76}(x^3 + x + 1)(x^3 + x^2 + 1)  \nonumber \\
&& (x^{12} + x^{11} + x^8 + x^6 + x^4 + x^3 + x^2 + x + 1)^3 \nonumber \\
&& (x^{12} + x^{11} + x^{10} + x^9 + x^8 + x^6 + x^4 + x + 1)^3 =0.
\end{eqnarray}
Then we have $x^2 + x + 1=0$, $x^3 + x + 1=0$, $x^3 + x^2 + 1=0$, $x^{12} + x^{11} + x^8 + x^6 + x^4 + x^3 + x^2 + x + 1=0$ or $x^{12} + x^{11} + x^{10} + x^9 + x^8 + x^6 + x^4 + x + 1=0$.  Thereby
$x\in \mathbb{F}_{2^2}$, $x\in \mathbb{F}_{2^3}$ or $x\in \mathbb{F}_{2^{12}}$.

Assume that $x\in \mathbb{F}_{2^2}$. When $k$ is odd, we have $x\in\mathbb{F}_{2^2}\cap\mathbb{F}_{2^n}=\mathbb{F}_{2}$, which contradicts with $x\neq 0,1$. When $k$ is even, we have $x\in\mathbb{F}_{2^2}\cap\mathbb{F}_{2^n}=\mathbb{F}_{2^2}$. Then $x^{2^{2k}}=x$ and  $x^{2^k}=x$. Thus we derive from Eq. \eqref{9.2} that $x^8+x=0$. This means $x\in \mathbb{F}_{2^3}$, and $x\in\mathbb{F}_{2^2}\cap\mathbb{F}_{2^3}=\mathbb{F}_{2}$, which is a contradiction.

Assume that $x\in \mathbb{F}_{2^3}$. When $k\equiv 1~({\rm mod}\ 3)$, we have $x^{2^{2k}}=x^4$ and  $x^{2^k}=x^2$. It follows from Eq. \eqref{9.2} that
\begin{eqnarray*}
x^{6} + x^{4} + x^5 + x^3 = x^3(x+1)^3=0.
\end{eqnarray*}
Thus we have $x\in \mathbb{F}_{2}$, it is impossible since $x\neq 0,1$.
When $k\equiv 2~({\rm mod}\ 3)$, we have $x^{2^{2k}}=x^2$ and  $x^{2^k}=x^4$. We conclude from Eq. \eqref{9.2} that
\begin{eqnarray*}
x^{3} + x^{2} + x^6 + x^5 = x^2(x+1)(x^2 + x + 1)=0.
\end{eqnarray*}
Observe that $x^2 + x + 1$ is an irreducible polynomial in $\mathbb{F}_{2}$. We have $x\in \mathbb{F}_{2^2}$. It leads to $x\in\mathbb{F}_{2^2}\cap\mathbb{F}_{2^3}=\mathbb{F}_{2}$, which contradicts with $x\neq 0,1$.

Assume that $x\in \mathbb{F}_{2^{12}}$. When $k\equiv 1~({\rm mod}\ 4)$ or $k\equiv 3~({\rm mod}\ 4)$, we have $\mathbb{F}_{2^{12}}\cap\mathbb{F}_{2^n}=\mathbb{F}_{2^3}$. It means that the solutions of Eq. \eqref{9.5} is in $\mathbb{F}_{2^3}$, which is impossible since $x\not\in \mathbb{F}_{2^3}$.

When $k\equiv 2~({\rm mod}\ 4)$, we have $x\in\mathbb{F}_{2^{12}}\cap\mathbb{F}_{2^n}=\mathbb{F}_{2^6}$. If $k\equiv 1~({\rm mod}\ 3)$, then $x^{2^{2k}}=x^4$ and  $x^{2^k}=x^{16}$. We derive from Eq. \eqref{9.2} that $x^{6} + x^{4} + x^{68} + x^{66} =0$. With the help of MAGMA, we can decompose it into the product of irreducible factors in $\mathbb{F}_{2}$ as
\begin{eqnarray*}
 && x^4(x+1)^4(x^5 + x^2 + 1)^2(x^5 + x^3 + 1)^2(x^5 + x^3 + x^2 + x + 1)^2(x^5 + x^4 + x^2 + x + 1)^2    \nonumber \\
&& (x^5 + x^4 + x^3 + x + 1)^2(x^5 + x^4 + x^3 + x^2 + 1)^2=0.
\end{eqnarray*}
The solutions of the above equation are in $\mathbb{F}_{2^5}$. Then $x\in\mathbb{F}_{2^6}\cap\mathbb{F}_{2^5}=\mathbb{F}_{2}$, which contradicts with $x\neq 0,1$.
If $k\equiv 2~({\rm mod}\ 3)$, then $x^{2^{2k}}=x^{16}$ and  $x^{2^k}=x^4$. We derive from Eq. \eqref{9.2} that
\begin{eqnarray*}
x^{24} + x^{16} + x^{20} + x^{12}= x^{12}(x+1)^{12}=0,
\end{eqnarray*}
which means $x\in\mathbb{F}_{2}$. It leads to a contradiction.

When $k\equiv 0~({\rm mod}\ 4)$, we have $\mathbb{F}_{2^{12}}\cap\mathbb{F}_{2^n}=\mathbb{F}_{2^{12}}$. If $k\equiv 1~({\rm mod}\ 3)$, then $x^{2^{2k}}=x^{256}$ and  $x^{2^k}=x^{16}$. We conclude from Eq. \eqref{9.2} that $x^{288} + x^{258} + x^{256} + x^{68} + x^{66} + x^{36} =0$. With the help of MAGMA, it can be decomposed into the following product of irreducible factors in $\mathbb{F}_{2}$ as
\begin{eqnarray*}
 && x^{36}(x+1)^{36}(x^{54} + x^{53} + x^{52} + x^{51} + x^{50} + x^{49} + x^{46} + x^{45} + x^{44} + x^{40} + x^{39} + x^{38} + x^{36} \nonumber \\
&&+ x^{35} + x^{30} + x^{29} + x^{27} + x^{25} + x^{24} + x^{20} + x^{19} + x^{18} + x^{17} + x^{15} + x^{13} + x^{12} + x^{10} + x^9 \nonumber \\
&&+ x^8 + x^5 + x^4 + x^3 + x^2 + x + 1)^2(x^{54} + x^{53} + x^{52} + x^{51} + x^{50} + x^{49} + x^{46} + x^{45} + x^{44} \nonumber \\
&&+ x^{42} + x^{41} + x^{39} + x^{37} + x^{36} + x^{35} + x^{34} + x^{30} + x^{29} + x^{27} + x^{25} + x^{24} + x^{19} + x^{18} + x^{16} \nonumber \\
&&+ x^{15} + x^{14} + x^{10} + x^9 + x^8 + x^5 + x^4 + x^3 + x^2 + x + 1)^2 =0.
\end{eqnarray*}
The solutions of the equation are in $\mathbb{F}_{2^{54}}$. Observe that $\mathbb{F}_{2^{12}}\cap\mathbb{F}_{2^{54}}=\mathbb{F}_{2^6}$, we obtain $x^{2^{2k}}=x^{4}$ and  $x^{2^k}=x^{16}$. It follows from Eq. \eqref{9.2} that $x^6 + x^4 + x^5 + x^3 =x^3(x+1)^3=0$. Thus $x\in \mathbb{F}_{2}$, which is a contradiction.
If $k\equiv 2~({\rm mod}\ 3)$, then $x^{2^{2k}}=x^{16}$ and  $x^{2^k}=x^{256}$. We derive from Eq. \eqref{9.2} that  $x^{528} + x^{18} + x^{16} + x^{1028} + x^{1026} + x^{516} =0$. Similarly, by MAGMA computation, we can decompose it into the following product of irreducible factors in $\mathbb{F}_{2}$ as
\begin{eqnarray}\label{9.9}
  x^{16}(x+1)^{16}(x^{83} + \cdots + 1)^2(x^{83} + \cdots + 1)^2(x^{83} + \cdots + 1)^2(x^{83} + \cdots + 1)^2 =0,
\end{eqnarray}
where the four irreducible polynomials over $\mathbb{F}_{2}$ can be seen in Appendix. Hence we obtain that $x\in \mathbb{F}_{2^{83}}$. It leads to $x\in\mathbb{F}_{2^{12}}\cap\mathbb{F}_{2^{83}}=\mathbb{F}_{2}$,
which is a contradiction. Therefore, Eq. \eqref{9.0} has no solution in $\mathbb{F}_{2^n}\backslash\{0,1\}$. This completes the proof.
\end{proof}

\begin{theorem}\label{th:2.10}
Let $n$ be an integer and $k$ be odd with $n=3k$. Then
\begin{eqnarray*}
f(x)= x^{2^{2k-1}+2^{k}+1}
\end{eqnarray*}
is a $0$-APN function over $\mathbb{F}_{2^n}$.
\end{theorem}

\begin{proof}
We need to prove that the equation
\begin{eqnarray}\label{10.0}
(x + 1)^{2^{2k-1}+2^{k}+1} + x^{2^{2k-1}+2^{k}+1} + 1 = 0
\end{eqnarray}
has no solution in $\mathbb{F}_{2^n}\backslash\{0,1\}$. Eq. \eqref{10.0}  can be simplified as
\begin{eqnarray}\label{10.1}
x^{2^{2k-1}+2^{k}} + x^{2^{2k-1}+1} + x^{2^{k}+1} + x^{2^{2k-1}} + x^{2^{k}} + x = 0.
\end{eqnarray}
Raising the square to Eq. \eqref{10.1} leads to
\begin{eqnarray}\label{10.2}
x^{2^{2k}+2^{k+1}} + x^{2^{2k}+2} + x^{2^{k+1}+2} + x^{2^{2k}} + x^{2^{k+1}} + x^2 = 0.
\end{eqnarray}
Let $y=x^{2^k}$ and $z=y^{2^k}$. Then $z^{2^k}=x$, and raising the $2^k$-th and $2^{2k}$-th power to Eq. \eqref{10.2} respectively we have
\begin{subequations}  \label{eq:10.3}
\begin{empheq}[left=\empheqlbrace]{align}
y^2z + x^2z + x^2y^2 + z + y^2 + x^2 = 0,  \label{10.3a} \\
z^2x + y^2x + y^2z^2 + x + z^2 + y^2 = 0,  \label{10.3b} \\
x^2y + z^2y + z^2x^2 + y + x^2 + z^2 = 0.  \label{10.3c}
\end{empheq}
\end{subequations}
Computing the resultants of Eq. \eqref{10.3a} and Eq. \eqref{10.3b}, Eq. \eqref{10.3a} and Eq. \eqref{10.3c} with respect to $z$ respectively, we obtain
\begin{subequations}  \label{eq:10.4}
\begin{empheq}[left=\empheqlbrace]{align}
(x^2 + x + 1)(x + y^2)(xy^2 + y^2 + 1)(xy^2 + x + 1) = 0,  \label{10.4a} \\
(y^2 + y + 1)(x^2 + y)(x^2y + y + 1)(x^2y + x^2 + 1) = 0.  \label{10.4b}
\end{empheq}
\end{subequations}
Suppose $x^2 + x + 1=0$. We get $x \in\mathbb{F}_{2^2}$, and thus $x \in\mathbb{F}_{2^2} \cap\mathbb{F}_{2^n}=\mathbb{F}_{2}$ for $k$ being odd, which contradicts with $x\neq0,1$.
Therefore, $x^2 + x + 1\neq0$. Similarly, it can be checked that $y^2 + y + 1\neq0$.

We compute the resultant of Eq. \eqref{10.4a}$/(x^2 + x + 1)$ and Eq. \eqref{10.4b}$/(y^2 + y + 1)$ with respect to $y$, and then decomposing it into the following product of irreducible factors over $\mathbb{F}_{2}$ as
\begin{eqnarray}\label{10.5}
 x^3(x + 1)^3(x^2 + x + 1)^{9}(x^3 + x + 1)^3(x^3 + x^2 + 1)^3 =0.
\end{eqnarray}
Notice that $x\neq0,1$, hence we have $x^2 + x + 1=0$, $x^3 + x + 1=0$ or $x^3 + x^2 + 1=0$. Suppose $x^2 + x + 1=0$. Then $x\in\mathbb{F}_{2^2}$, and hence
$x\in\mathbb{F}_{2^2}\cap\mathbb{F}_{2^n}=\mathbb{F}_{2}$, which is a contradiction.

Suppose $x^3 + x + 1=0$ or $x^3 + x^2 + 1=0$. We get $x\in\mathbb{F}_{2^3}$.
When $k\equiv 0~({\rm mod}\ 3)$, we have $x^{2^{2k}}=x$ and $x^{2^{k+1}}=x^2$. It follows from Eq. \eqref{10.2} that $$ x^4+x=0.$$
This yields $x=0$ or $x^3=1$. If $x^3=1$, from $x^3 + x + 1=0$ or $x^3 + x^2 + 1=0$ we obtain $x=0$, which is a contradiction.

When $k\equiv 1~({\rm mod}\ 3)$, we get $x^{2^{2k}}=x^4$ and $x^{2^{k+1}}=x^4$. We conclude from Eq. \eqref{10.2} that $$ x^2(x^3+1)^2=0.$$
Thus we have $x=0$ or $x^3=1$. Similarly, from $x^3 + x + 1=0$ or $x^3 + x^2 + 1=0$ we obtain $x=0$, which is impossible.

When $k\equiv 2~({\rm mod}\ 3)$, we obtain $x^{2^{2k}}=x^2$ and $x^{2^{k+1}}=x$. We deduce from Eq. \eqref{10.2} that $$ x(x^3+1)=0.$$
Thus we have $x=0$ or $x^3=1$. Similarly, we only obtain a solution $x=0$. It leads to a contradiction.
Hence, Eq. \eqref{10.0} has no solution in $\mathbb{F}_{2^n}\backslash\{0,1\}$.
The proof is completed.
\end{proof}

\begin{theorem}\label{th:2.11}
Let $n$ be a positive integer and $k$ be odd with $n=3k$. Then
\begin{eqnarray*}
f(x)= x^{2^{2k}-2^{k+1}-1}
\end{eqnarray*}
is a $0$-APN function over $\mathbb{F}_{2^n}$.
\end{theorem}

\begin{proof}
It suffices to prove that the equation
\begin{eqnarray}\label{11.0}
(x + 1)^{2^{2k}-2^{k+1}-1} + x^{2^{2k}-2^{k+1}-1} + 1 = 0
\end{eqnarray}
has no solution in $\mathbb{F}_{2^n}\backslash\{0,1\}$. We derive from Eq. \eqref{11.0} that
\begin{eqnarray}\label{11.1}
x^{2^{2k}+2^{k+1}} + x^{2^{2k}+1} + x^{2^{2k}} + x^{2^{k+2}+2} + x^{2^{k+2}+1} + x^{2^{k+1}+2} = 0.
\end{eqnarray}
Let $y=x^{2^k}$ and $z=y^{2^k}$. Then $z^{2^k}=x$, and raising the $2^k$-th and $2^{2k}$-th power to Eq. \eqref{11.1} respectively obtains
\begin{subequations}  \label{eq:11.3}
\begin{empheq}[left=\empheqlbrace]{align}
y^2z + xz + z + x^2y^4 + xy^4 + x^2y^2 = 0,  \label{11.3a} \\
z^2x + yx + x + y^2z^4 + yz^4 + y^2z^2 = 0,  \label{11.3b} \\
x^2y + zy + y + z^2x^4 + zx^4 + z^2x^2 = 0.  \label{11.3c}
\end{empheq}
\end{subequations}
Computing the resultants of Eq. \eqref{11.3a} and Eq. \eqref{11.3b}, Eq. \eqref{11.3a} and Eq. \eqref{11.3c} with respect to $z$ respectively, we have
\begin{subequations}  \label{eq:11.4}
\begin{empheq}[left=\empheqlbrace]{align}
& x^8y^{18} + x^8y^{17} + x^8y^{10} + x^8y^9 + x^7y^8 + x^7y^4 + x^6y^{10} + x^6y^6 + x^5y^{12} + x^5y^8 + x^5y^4 \notag \\
& + x^5y + x^5 + x^4y^{18} + x^4y^{17} + x^4y^{14} + x^4y^{10} + x^4y^6 + x^3y^{12} + x^3y^8 + x^2y^{14} + x^2y^{10} \notag \\
& + xy^9 + xy^8 + xy + x = 0,  \label{11.4a} \\
& x^8y^8 + x^8y^4 + x^7y^4 + x^7y^2 + x^6y^6 + x^6y^2 + x^5y^6 + x^5y^4 + x^4y^8 + x^4y + x^3y^5 + x^3y^3 \notag \\
& + x^2y^7 + x^2y^3 + xy^7 + xy^5 + y^5 + y= 0.  \label{11.4b}
\end{empheq}
\end{subequations}
Computing the resultant of Eq. \eqref{11.4a} and Eq. \eqref{11.4b} with respect to $y$, by MAGMA computation, and then decomposing it into
\begin{eqnarray*}
 x^{31}(x + 1)^{31}(x^2 + x + 1)^{60} =0.
\end{eqnarray*}
Note that $x^2 + x + 1 $ is irreducible on $\mathbb{F}_{2}$. Suppose $x^2 + x + 1=0$. Then $x \in\mathbb{F}_{2^2} \cap\mathbb{F}_{2^n}=\mathbb{F}_{2}$ for $k$ being odd, which contradicts with $x\neq0,1$. Hence Eq. \eqref{11.0} has no solution in $\mathbb{F}_{2^n}\backslash\{0,1\}$.
The proof is completed.
\end{proof}

\begin{theorem}\label{th:2.18}
Let $n$ and $k$ be positive integers with $n=3k$. Then
\begin{eqnarray*}
f(x)= x^{2^{2k+1}-2^{k}-1}
\end{eqnarray*}
is a $0$-APN function over $\mathbb{F}_{2^n}$.
\end{theorem}

\begin{proof}
We need to prove that the equation
\begin{eqnarray}\label{18.0}
(x + 1)^{2^{2k+1}-2^{k}-1} + x^{2^{2k+1}-2^{k}-1} + 1 = 0
\end{eqnarray}
has no solution in $\mathbb{F}_{2^n}\backslash\{0,1\}$. Assume that $x\neq0,1$ is a solution of Eq. \eqref{18.0}.  Eq. \eqref{18.0} can be simplified as
\begin{eqnarray}\label{18.1}
x^{2^{2k+1}+2^{k}} + x^{2^{2k+1}+1} + x^{2^{2k+1}} + x^{2^{k+1}+2}  + x^{2^{k+1}+1} + x^{2^{k}+2} = 0.
\end{eqnarray}
Let $y=x^{2^k}$ and $z=y^{2^k}$. Then $z^{2^k}=x$, and taking the $2^k$-th and $2^{2k}$-th power to Eq. \eqref{18.1} respectively derives
\begin{subequations}  \label{eq:18.3}
\begin{empheq}[left=\empheqlbrace]{align}
yz^2 + xz^2 + z^2 + x^2y^2 + xy^2 + x^2y = 0,  \label{18.3a} \\
zx^2 + yx^2 + x^2 + y^2z^2 + yz^2 + y^2z = 0,  \label{18.3b} \\
xy^2 + zy^2 + y^2 + z^2x^2 + zx^2 + z^2x = 0.  \label{18.3c}
\end{empheq}
\end{subequations}
Computing the resultants of Eq. \eqref{18.3a} and Eq. \eqref{18.3b}, Eq. \eqref{18.3a} and Eq. \eqref{18.3c} with respect to $z$ respectively obtains
\begin{eqnarray*}
Res_{1}(x,y) &=& x(x+1)(x + y^2)(x^4y^2 + x^4y + x^3y^4 + x^3y^2 + x^3 + x^2y^6 + x^2y^4 + x^2y^2 + x^2 \nonumber \\
&&+ xy^6  + xy^4 + xy^2 + y^5 + y^4) = 0,   \\ 
Res_{2}(x,y) &=& y(y+1)(x^2 + y)(x^6y^2 + x^6y + x^5 + x^4y^3 + x^4y^2 + x^4y + x^4 + x^2y^4 + x^2y^3 \nonumber \\
&&+ x^2y^2 + x^2y + xy^4 + y^3 + y^2) = 0.  
\end{eqnarray*}
Notice that $x,y\not\in\mathbb{F}_{2}$, we compute the resultant of  $Res_{1}(x,y)/x(x+1)$ and $Res_{2}(x,y)/y(y+1)$ with respect to $y$, and then we decompose the resultant into
\begin{eqnarray*}
x^{19}(x + 1)^{19}(x^2 + x + 1)^{16} =0.
\end{eqnarray*}
It is obvious that the polynomial $x^2 + x + 1$ is irreducible on $\mathbb{F}_{2}$. Suppose $x^2 + x + 1=0$, this yields $x\in \mathbb{F}_{2^2}$.   When $k$ is odd, we have $x\in \mathbb{F}_{2^2}\cap\mathbb{F}_{2^n}=\mathbb{F}_{2}$, which contradicts with $x\neq0,1$. When $k$ is even, we have $x\in \mathbb{F}_{2^2}\cap\mathbb{F}_{2^n}=\mathbb{F}_{2^2}$. Then $x^{2^{2k}}=x$ and  $x^{2^k}=x$. Thus we derive from Eq. \eqref{18.1} that $x^4+x^2=0$. This means $x\in \mathbb{F}_{2}$, which is a contradiction.
Therefore, Eq. \eqref{18.0} has no solution in $\mathbb{F}_{2^n}\backslash\{0,1\}$. This completes the proof.
\end{proof}

\subsection{The case of $n=3k+1$}
In this subsection, we present a new class of $0$-APN power functions over $\mathbb{F}_{2^n}$ with $n=3k+1$.

\begin{theorem}\label{th:2.20}
Let $n$ and $k$ be positive integers with $k\not\equiv 11~({\rm mod}\ 34)$ and $n=3k+1$. Then
\begin{eqnarray*}
f(x)= x^{3\cdot(2^{k+1}-1)}
\end{eqnarray*}
is a $0$-APN function over $\mathbb{F}_{2^n}$.
\end{theorem}

\begin{proof}
It suffices to prove that the equation
\begin{eqnarray}\label{20.0}
(x + 1)^{3\cdot(2^{k+1}-1)} + x^{3\cdot(2^{k+1}-1)} + 1 = 0
\end{eqnarray}
has no solution in $\mathbb{F}_{2^n}\backslash\{0,1\}$. Eq. \eqref{20.0} can be written as
\begin{eqnarray}\label{20.1}
x^{2^{k+2}+3} + x^{2^{k+1}+3} + x^{3\cdot2^{k+1}+2} + x^{3\cdot2^{k+1}+1} + x^{3\cdot2^{k+1}} + x^6 + x^5 + x^4 = 0.
\end{eqnarray}
Let $y=x^{2^k}$ and $z=y^{2^k}$. Then $z^{2^{k+1}}=x$, and raising the $2^{k}$-th and $2^{2k}$-th power to Eq. \eqref{20.1} respectively derives
\begin{subequations}  \label{eq:20.2}
\begin{empheq}[left=\empheqlbrace]{align}
x^3y^4 + x^3y^2 + x^2y^6 + xy^6 + y^6 + x^6 + x^5 + x^4 = 0,  \label{20.2a} \\
y^3z^4 + y^3z^2 + y^2z^6 + yz^6 + z^6 + y^6 + y^5 + y^4 = 0,  \label{20.2b} \\
z^3x^2 + z^3x   + z^2x^3 + zx^3 + x^3 + z^6 + z^5 + z^4 = 0.  \label{20.2c}
\end{empheq}
\end{subequations}
Computing the resultants of Eq. \eqref{20.2b} and Eq. \eqref{20.2c} with respect to $z$ respectively, we obtain
\begin{eqnarray} \label{20.3}
&& x^{18}y^{20} + x^{18}y^{18} + x^{18}y^{12} + x^{18}y^{10} + x^{18}y^8 + x^{18}y^2 + x^{18} + x^{16}y^{21} + x^{16}y^{19} + x^{16}y^{18} + x^{16}y^{17} \nonumber \\
&& + x^{16}y^{13} + x^{16}y^{11} + x^{16}y^9 + x^{16}y^8 + x^{16}y^4 + x^{14}y^{16} + x^{14}y^4 + x^{12}y^{20} + x^{12}y^{16} + x^{10}y^{26} \nonumber \\
&&+ x^{10}y^{25} + x^{10}y^{24} + x^{10}y^{22} + x^{10}y^{21} + x^{10}y^{16} + x^{10}y^{14} + x^{10}y^{13} + x^{10}y^{12} + x^{10}y^{10} + x^{10}y^9 \nonumber \\
&& + x^{10}y^8 + x^8y^{28} + x^8y^{27} + x^8y^{26} + x^8y^{24} + x^8y^{23} + x^8y^{22} + x^8y^{20} + x^8y^{15} + x^8y^{14} + x^8y^{12} \nonumber \\
&&+ x^8y^{11} + x^8y^{10} + x^6y^{20} + x^6y^{16} + x^4y^{32} + x^4y^{20} + x^2y^{32} + x^2y^{28} + x^2y^{27} + x^2y^{25} + x^2y^{23} \nonumber \\
&&+ x^2y^{19} + x^2y^{18} + x^2y^{17} + x^2y^{15} + y^{36} + y^{34} + y^{28} + y^{26} + y^{24} + y^{18} + y^{16} = 0.
\end{eqnarray}
Computing the resultant of Eq. \eqref{20.2a} and Eq. \eqref{20.3} with respect to $y$, with the help of MAGMA, and then decomposing the resultant into the following product of irreducible factors on $\mathbb{F}_{2}$ as
\begin{eqnarray}\label{20.4}
&& x^{64}(x + 1)^{64}(x^5 + x^2 + 1)(x^5 + x^3 + 1)(x^5 + x^3 + x^2 + x + 1)(x^5 + x^4 + x^2 + x + 1)    \nonumber \\
&& (x^5 + x^4 + x^3 + x + 1)(x^5 + x^4 + x^3 + x^2 + 1) (x^{34} + x^{30} + x^{29} + x^{28} + x^{27} + x^{23} \nonumber \\
&&+ x^{22} + x^{21} + x^{20} + x^{17} + x^{14} + x^{13} + x^{12} + x^{11} + x^7 + x^6 + x^5 + x^4 + 1)(x^{34} + x^{32} \nonumber \\
&&+ x^{30} + x^{29} + x^{28} + x^{27} + x^{23} + x^{22} + x^{21} + x^{20} + x^{17} + x^{16} + x^2 + x + 1)(x^{34} + x^{33} \nonumber \\
&&+ x^{32} + x^{18} + x^{17} + x^{14} + x^{13} + x^{12} + x^{11} + x^7 + x^6 + x^5 + x^4 + x^2 + 1)=0.
\end{eqnarray}
Thus the solutions of Eq. \eqref{20.4} are in $\mathbb{F}_{2^5}$ or $\mathbb{F}_{2^{34}}$.  Suppose $x\in\mathbb{F}_{2^5}$.
When $k\not\equiv 3~({\rm mod}\ 5)$, we have $x\in \mathbb{F}_{2^5}\cap\mathbb{F}_{2^n}=\mathbb{F}_{2}$, which contradicts with $x\neq0,1$.
When $k\equiv 3~({\rm mod}\ 5)$, we have $x\in \mathbb{F}_{2^5}\cap\mathbb{F}_{2^n}=\mathbb{F}_{2^5}$. Thereby $x^{2^{2k}}=x^2$ and  $x^{2^k}=x^8$. It follows from Eq. \eqref{20.1} that
\begin{eqnarray*}
x^{18} + x^{17} + x^6 + x^5 = x^5(x+1)^5(x^2 + x + 1)^4=0,
\end{eqnarray*}
which means $x\in \mathbb{F}_{2^2}$ since $x^2 + x + 1$ is irreducible in $\mathbb{F}_{2}$. This yields $x\in \mathbb{F}_{2^{5}}\cap\mathbb{F}_{2^2}=\mathbb{F}_{2}$, which is a contradiction.

Suppose $x\in \mathbb{F}_{2^{34}}$. When $k$ is even, we have $x\in \mathbb{F}_{2^{34}}\cap\mathbb{F}_{2^n}=\mathbb{F}_{2}$, which is a contradiction.
When $k$ is odd, we have $x\in \mathbb{F}_{2^{34}}\cap\mathbb{F}_{2^n}=\mathbb{F}_{2^2}$ since $k\not\equiv 11~({\rm mod}\ 34)$. Hence we have $x^{2^{2k}}=x$ and  $x^{2^k}=x^2$. We derive from Eq. \eqref{20.1} that  
\begin{eqnarray*}
 x^{11} + x^{7} + x^{14} + x^{13} + x^{12} + x^6 + x^5 + x^4 = x^4(x+1)^4(x^3+x+1)(x^3+x^2+1)=0.
\end{eqnarray*}
It can be checked that the polynomials $x^3+x+1$ and $x^3+x^2+1$ are irreducible in $\mathbb{F}_{2}$. Thus the solutions of the above equation are in $\mathbb{F}_{2^3}$, which implies that $x\in \mathbb{F}_{2^2}\cap\mathbb{F}_{2^3}=\mathbb{F}_{2}$. It leads to a contradiction.
Therefore, Eq. \eqref{20.0} has no solution in $\mathbb{F}_{2^n}\backslash\{0,1\}$.
This completes the proof.
\end{proof}

From the examples in Table \ref{tab} and by MAGMA computation, according to Lemma \ref{Lem:CCZ}, it is easy to see that all these $0$-APN power functions over $\mathbb{F}_{2^n}$ are CCZ-inequivalent to each other for $6\leq n\leq11$.

\section{Two classes of $0$-APN power functions over $\mathbb{F}_{2^n}$}\label{fourth}
In this section, we completely characterize two new classes of $0$-APN power functions over $\mathbb{F}_{2^n}$ by studying some special equations.

\begin{theorem}\label{th:4.1}
Let $n$, $m$, $k$ and $d$ be positive integers with $(2^k-1)d\equiv 2^m-1~({\rm mod}\ 2^n-1)$. Then
\begin{eqnarray*}
f(x)= x^d
\end{eqnarray*}
is a $0$-APN function over $\mathbb{F}_{2^n}$ if and only if $\gcd(n, m)=\gcd(n, m-k)=1$.
\end{theorem}

\begin{proof}
We consider the equation
\begin{eqnarray}\label{eq:4.0}
(x + 1)^d + x^d + 1 = 0.
\end{eqnarray}
Since $(2^k-1)d\equiv 2^m-1~({\rm mod}\ 2^n-1)$, we have $2^k\cdot d+1\equiv 2^m+d~({\rm mod}\ 2^n-1)$. It follows from
Eq. \eqref{eq:4.0}  that
\begin{eqnarray*}
(x + 1)^{2^m-1} = (x + 1)^{(2^k-1)d}= (x^d + 1)^{2^k-1}.
\end{eqnarray*}
It can be rewritten as
\begin{eqnarray}\label{eq:4.1}
x^{2^m+d} + x^d +x^{2^m} + x^{2^k\cdot d+1} + x^{d\cdot2^k} + x =0,
\end{eqnarray}
which implies
\begin{eqnarray*}
x^{2^k\cdot d} + x^d = x^{2^m} + x.
\end{eqnarray*}
This yields
\begin{eqnarray*}
x^d(x^{(2^k-1)\cdot d}+1) = x(x^{2^m-1} + 1).
\end{eqnarray*}
We further have
\begin{eqnarray}\label{eq:4.3}
x(x^{d-1}+1)(x^{2^m-1} + 1)=0.
\end{eqnarray}
Equation \eqref{eq:4.3} has only two solutions $x=0$ and $x=1$ if and only if $\gcd(2^m-1, 2^n-1)=1$ and $\gcd(d-1, 2^n-1)=1$.
It is easy to check that  $\gcd(2^m-1, 2^n-1)=1$ if and only if $\gcd(m, n)=1$.
Since $(2^k-1)d\equiv 2^m-1~({\rm mod}\ 2^n-1)$, we have $\gcd(k, n)=1$ and $\gcd(2^k-1, 2^n-1)=1$.
Note that
$$(2^k-1)(d-1)\equiv 2^k(2^{m-k}-1)~({\rm mod}\ 2^n-1).$$
%
This yields
$\gcd(d-1, 2^n-1)=\gcd(2^{m-k}-1, 2^n-1)=1$ if and only if $\gcd(m-k, n)=1$. Therefore we conclude the desire conclusion, and this completes the proof.
\end{proof}

\begin{lemma}\label{lemma:4.2}
Let $n$, $m$, $k$ and $d$ be positive integers with $(2^k+1)d\equiv 2^m+1~({\rm mod}\ 2^n-1)$. Then $\gcd(d\cdot2^k-1, 2^n-1)=1$ and $\gcd(d-1, 2^n-1)=1$
if and only if one of the following holds:

\emph{(i)} $\frac{n}{\gcd(n, k)}$ is odd and $\gcd(n, m+k)=\gcd(n, m-k)=1$.

\emph{(ii)} $d\equiv 0~({\rm mod}\ 3)$, $n$ is even, $k$ and $m$ are odd satisfying that $\frac{n}{\gcd(n, k)}$ is even, $\gcd(k, n)=1$, and $\gcd(m+k,n)= \gcd(m-k,n)=2$.
\end{lemma}

\begin{proof}
Since
\begin{eqnarray}\label{eq:6.1}
(2^k+1)d\equiv 2^m+1~({\rm mod}\ 2^n-1),
\end{eqnarray}
we obtain
\begin{eqnarray*}
d\cdot 2^k -1\equiv 2^m-d~({\rm mod}\ 2^n-1).
\end{eqnarray*}
Multiplying the equation with $2^k+1$ and then plugging Eq. \eqref{eq:6.1} into it gives
\begin{eqnarray}\label{eq:6.3}
(2^k+1)(d\cdot 2^k -1)\equiv (2^k+1)(2^m-d)\equiv 2^{m+k}+2^m - (2^m+1)\equiv 2^{m+k}-1~({\rm mod}\ 2^n-1).
\end{eqnarray}
We further derive from Eq. \eqref{eq:6.1} that
\begin{eqnarray}\label{eq:6.4}
(2^k+1)(d-1)\equiv 2^m+1 - (2^k+1)\equiv 2^k(2^{m-k}-1)~({\rm mod}\ 2^n-1).
\end{eqnarray}

Case 1: If $\frac{n}{\gcd(n, k)}$ is odd, then $\gcd(2^k+1, 2^n-1)=1$. Therefore, we conclude from Eqs. \eqref{eq:6.3} and \eqref{eq:6.4} that  $\gcd(d\cdot2^k-1, 2^n-1)=1$ and $\gcd(d-1, 2^n-1)=1$ if and only if $\gcd(2^{m+k}-1, 2^n-1)=1$ and $\gcd(2^{m-k}-1, 2^n-1)=1$, this yields  $\gcd(n, m+k)=1$ and $\gcd(n, m-k)=1$.

Case 2: If $\frac{n}{\gcd(n, k)}$ is even, then $n$ is even and $$\gcd(2^k+1, 2^n-1)=\frac{2^{\gcd(2k, n)}-1}{2^{\gcd(k, n)}-1}=2^{\gcd(k, n)}+1.$$
Observe that $\gcd(2^{m+k}-1, 2^n-1)=2^{\gcd(m+k,n)}-1$ and $\gcd(2^k(2^{m-k}-1), 2^n-1)=2^{\gcd(m-k,n)}-1$.

From $\gcd(d\cdot2^k-1, 2^n-1)=1$ and $\gcd(d-1, 2^n-1)=1$, we have $$2^{\gcd(k, n)}+1 = 2^{\gcd(m+k,n)}-1 = 2^{\gcd(m-k,n)}-1.$$
The above equation holds if and only if $\gcd(k, n)=1$ and $\gcd(m+k,n)= \gcd(m-k,n)=2$. This means that
$k$ is odd, $m$ is odd and $2^{\gcd(k, n)}+1 =3$. 
We also have $3\nmid(d-1)$ and $3\nmid(2^k\cdot d-1)$, which is equivalent to $d\not\equiv 1~({\rm mod}\ 3)$ and $d\not\equiv 2~({\rm mod}\ 3)$, i.e.,
$d\equiv 0~({\rm mod}\ 3)$.
\end{proof}

\begin{theorem}\label{th:4.2}
Let $n$, $m$, $k$ and $d$ be positive integers satisfying $(2^k+1)d\equiv 2^m+1~({\rm mod}\ 2^n-1)$. Then
\begin{eqnarray*}
f(x)= x^d
\end{eqnarray*}
is a $0$-APN function over $\mathbb{F}_{2^n}$ if and only if one of the following holds:

\emph{(i)} $\frac{n}{\gcd(n, k)}$ is odd and $\gcd(n, m+k)=\gcd(n, m-k)=1$.

\emph{(ii)} $d\equiv 0~({\rm mod}\ 3)$, $n$ is even, $k$ and $m$ are odd satisfying that $\frac{n}{\gcd(n, k)}$ is even, $\gcd(k, n)=1$, and $\gcd(m+k,n)= \gcd(m-k,n)=2$.
\end{theorem}

\begin{proof}
We consider the equation
\begin{eqnarray}\label{eq:5.0}
(x + 1)^d + x^d + 1 = 0.
\end{eqnarray}
Since $(2^k+1)d\equiv 2^m+1~({\rm mod}\ 2^n-1)$, we deduce $2^m-2^k\cdot d \equiv d-1~({\rm mod}\ 2^n-1)$. Therefore we derive from
Eq. \eqref{eq:5.0}  that
\begin{eqnarray*}
(x + 1)^{2^m+1} = (x + 1)^{(2^k+1)d}= (x^d + 1)^{2^k+1},
\end{eqnarray*}
which can be written as
\begin{eqnarray}\label{eq:5.1}
x^{(2^k+1)\cdot d} + x^{d\cdot2^k} + x^d + x^{2^m+1} + x^{2^m} + x =0.
\end{eqnarray}
Equation \eqref{eq:5.1} becomes
\begin{eqnarray*}
x(x^{d-1}+1) + x^{d\cdot2^k}(x^{d-1}+1) =0.
\end{eqnarray*}
This yields
\begin{eqnarray}\label{eq:5.3}
x(x^{d-1}+1) (x^{d\cdot2^k-1}+1) = 0.
\end{eqnarray}
Equation \eqref{eq:5.3} has only the two solutions $x=0,1$ if and only if  $\gcd(d\cdot2^k-1, 2^n-1)=1$ and $\gcd(d-1, 2^n-1)=1$.
According to Lemma \ref{lemma:4.2}, we derive the desire conclusion.
The proof is complete.
\end{proof}

According to the above theorems, we can deduce the following corollaries directly.

\begin{corollary}\label{cor:4.1}
Let $n$, $l$ and $k$ be positive integers. Then
\begin{eqnarray*}
f(x)= x^{\frac{2^{lk}-1}{2^k-1}}
\end{eqnarray*}
is a $0$-APN function over $\mathbb{F}_{2^n}$ if and only if $\gcd(n, lk)=\gcd(n, (l-1)k)=1$.
\end{corollary}

\begin{corollary}\label{cor:4.2}
Let $n$, $l$ be odd and $k$ be a positive integer with $\gcd(n, (l+1)k)=\gcd(n, (l-1)k)=1$. Then
\begin{eqnarray*}
f(x)= x^{\frac{2^{lk}+1}{2^k+1}}
\end{eqnarray*}
is a $0$-APN function over $\mathbb{F}_{2^n}$.
\end{corollary}

Theorems \ref{th:4.1} and \ref{th:4.2} present more new $0$-APN functions over $\mathbb{F}_{2^n}$.
Some of them are CCZ-inequivalent to the known ones. The corresponding examples can be seen in Table \ref{tab}.


\section{Conclusion}\label{conclu}

In this paper, we investigated the examples of exponents $d$ of the power function $f(x)=x^d$ over $\mathbb{F}_{2^n}$($1\leq n\leq11$) in Table $1$ of \cite{BKRS2020}. 
Based on the multivariate method and resultant elimination, we presented several new infinite classes of $0$-APN power functions over $\mathbb{F}_{2^n}$. Furthermore, two new classes of $0$-APN power functions over $\mathbb{F}_{2^n}$ were characterized completely by studying some special equations. These new $0$-APN power functions were CCZ-inequivalent to the known ones.


\section*{Appendix}\label{App}

In this section, we list some resultants by the computation of MAGMA used in the proofs of the results in the previous sections and some equations.

\vspace{0.6\baselineskip}
The following $ Res_{1}(x,y)$ and $Res_{2}(x,y)$ are the resultants of  Eq. \eqref{eq:1.5} in Theorem \ref{th1.1}.

$ Res_{1}(x,y) = x(x+1)(xy + y + 1)(xy + x + 1)(x^{28}y^{22} + x^{28}y^{21} + x^{28}y^{20} + x^{28}y^{19}+ x^{28}y^{14} + x^{28}y^{13} +
    x^{28}y^{12} + x^{28}y^{11} + x^{27}y^{20} + x^{27}y^{18} + x^{27}y^{12} + x^{27}y^{10} +
    x^{26}y^{22} + x^{26}y^{20} + x^{26}y^{18} + x^{26}y^{17} + x^{26}y^{16} + x^{26}y^{14} +
    x^{26}y^{12} + x^{26}y^{10} + x^{26}y^9 + x^{26}y^8 + x^{25}y^{21} + x^{25}y^{20} +
    x^{25}y^{19} + x^{25}y^{18} + x^{25}y^{16} + x^{25}y^{14} + x^{25}y^9 + x^{25}y^7 +
    x^{24}y^{22} + x^{24}y^{21} + x^{24}y^{18} + x^{24}y^{17} + x^{24}y^{15} + x^{24}y^{11} +
    x^{24}y^{10} + x^{24}y^6 + x^{23}y^{20} + x^{23}y^{17} + x^{23}y^{16} + x^{23}y^{12} +
    x^{23}y^{10} + x^{23}y^8 + x^{23}y^6 + x^{23}y^5 + x^{22}y^{22} + x^{22}y^{20} +
    x^{22}y^{19} + x^{22}y^{18} + x^{22}y^{17} + x^{22}y^{13} + x^{22}y^{11} + x^{22}y^{10} +
    x^{22}y^6 + x^{22}y^4 + x^{21}y^{21} + x^{21}y^{20} + x^{21}y^{18} + x^{21}y^{14} +
    x^{21}y^{12} + x^{21}y^9 + x^{20}y^{22} + x^{20}y^{21} + x^{20}y^{19} + x^{20}y^{18} +
    x^{20}y^{17} + x^{20}y^{16} + x^{20}y^{14} + x^{20}y^{12} + x^{20}y^{11} + x^{20}y^8 +
    x^{20}y^6 + x^{20}y^4 + x^{19}y^{20} + x^{19}y^{19} + x^{19}y^{17} + x^{19}y^{12} +
    x^{19}y^7 + x^{19}y^5 + x^{18}y^{22} + x^{18}y^{20} + x^{18}y^{18} + x^{18}y^{17} +
    x^{18}y^{15} + x^{18}y^{13} + x^{18}y^{11} + x^{18}y^{10} + x^{18}y^6 + x^{18}y^4 +
    x^{17}y^{21} + x^{17}y^{20} + x^{17}y^{19} + x^{17}y^{18} + x^{17}y^{16} + x^{17}y^{14} +
    x^{17}y^{13} + x^{17}y^{12} + x^{17}y^{11} + x^{17}y^{10} + x^{17}y^8 + x^{17}y^5 +
    x^{17}y^4 + x^{17}y^3 + x^{16}y^{22} + x^{16}y^{21} + x^{16}y^{18} + x^{16}y^{12} +
    x^{16}y^{11} + x^{16}y^8 + x^{16}y^6 + x^{16}y^2 + x^{15}y^{20} + x^{15}y^{17} +
    x^{15}y^{16} + x^{15}y^{10} + x^{15}y^9 + x^{15}y^8 + x^{15}y^6 + x^{15}y^4 + x^{15}y^2
    + x^{15}y + x^{14}y^{22} + x^{14}y^{20} + x^{14}y^{19} + x^{14}y^{18} + x^{14}y^{17} +
    x^{14}y^{16} + x^{14}y^{12} + x^{14}y^{11} + x^{14}y^{10} + x^{14}y^6 + x^{14}y^5 +
    x^{14}y^4 + x^{14}y^3 + x^{14}y^2 + x^{14} + x^{13}y^{21} + x^{13}y^{20} + x^{13}y^{18} +
    x^{13}y^{16} + x^{13}y^{14} + x^{13}y^{13} + x^{13}y^{12} + x^{13}y^6 + x^{13}y^5 +
    x^{13}y^2 + x^{12}y^{20} + x^{12}y^{16} + x^{12}y^{14} + x^{12}y^{11} + x^{12}y^{10} +
    x^{12}y^4 + x^{12}y + x^{12} + x^{11}y^{19} + x^{11}y^{18} + x^{11}y^{17} + x^{11}y^{14} +
    x^{11}y^{12} + x^{11}y^{11} + x^{11}y^{10} + x^{11}y^9 + x^{11}y^8 + x^{11}y^6 +
    x^{11}y^4 + x^{11}y^3 + x^{11}y^2 + x^{11}y + x^{10}y^{18} + x^{10}y^{16} + x^{10}y^{12}
    + x^{10}y^{11} + x^{10}y^9 + x^{10}y^7 + x^{10}y^5 + x^{10}y^4 + x^{10}y^2 + x^{10} +
    x^9y^{17} + x^9y^{15} + x^9y^{10} + x^9y^5 + x^9y^3 + x^9y^2 + x^8y^{18} +
    x^8y^{16} + x^8y^{14} + x^8y^{11} + x^8y^{10} + x^8y^8 + x^8y^6 + x^8y^5 +
    x^8y^4 + x^8y^3 + x^8y + x^8 + x^7y^{13} + x^7y^{10} + x^7y^8 + x^7y^4 +
    x^7y^2 + x^7y + x^6y^{18} + x^6y^{16} + x^6y^{12} + x^6y^{11} + x^6y^9 +
    x^6y^5 + x^6y^4 + x^6y^3 + x^6y^2 + x^6 + x^5y^{17} + x^5y^{16} + x^5y^{14}
    + x^5y^{12} + x^5y^{10} + x^5y^6 + x^5y^5 + x^5y^2 + x^4y^{16} + x^4y^{12} +
    x^4y^{11} + x^4y^7 + x^4y^5 + x^4y^4 + x^4y + x^4 + x^3y^{15} + x^3y^{13} +
    x^3y^8 + x^3y^6 + x^3y^4 + x^3y^3 + x^3y^2 + x^3y + x^2y^{14} +
    x^2y^{13} + x^2y^{12} + x^2y^{10} + x^2y^8 + x^2y^6 + x^2y^5 + x^2y^4 +
    x^2y^2 + x^2 + xy^{12} + xy^{10} + xy^4 + xy^2 + y^{11} + y^{10} + y^9 + y^8 +
    y^3 + y^2 + y + 1). $  \\

$ Res_{2}(x,y) = (xy + y + 1)(xy + x + 1)(x^{32}y^{28} + x^{32}y^{26} + x^{32}y^{24} + x^{32}y^{22} + x^{32}y^{20} + x^{32}y^{18} +
    x^{32}y^{16} + x^{32}y^{14} + x^{31}y^{26} + x^{31}y^{25} + x^{31}y^{22} + x^{31}y^{21} +
    x^{31}y^{18} + x^{31}y^{17} + x^{31}y^{14} + x^{31}y^{13} + x^{30}y^{27} + x^{30}y^{26} +
    x^{30}y^{24} + x^{30}y^{23} + x^{30}y^{22} + x^{30}y^{20} + x^{30}y^{19} + x^{30}y^{18} +
    x^{30}y^{16} + x^{30}y^{15} + x^{30}y^{14} + x^{30}y^{12} + x^{29}y^{27} + x^{29}y^{26} +
    x^{29}y^{23} + x^{29}y^{21} + x^{29}y^{20} + x^{29}y^{18} + x^{29}y^{15} + x^{29}y^{13} +
    x^{29}y^{12} + x^{29}y^{11} + x^{28}y^{27} + x^{28}y^{22} + x^{28}y^{21} + x^{28}y^{19} +
    x^{28}y^{18} + x^{28}y^{16} + x^{28}y^{14} + x^{28}y^{13} + x^{28}y^{10} + x^{28}y^8 +
    x^{27}y^{27} + x^{27}y^{26} + x^{27}y^{24} + x^{27}y^{23} + x^{27}y^{22} + x^{27}y^{20} +
    x^{27}y^{19} + x^{27}y^{17} + x^{27}y^{15} + x^{27}y^{13} + x^{27}y^{11} + x^{27}y^7 +
    x^{26}y^{27} + x^{26}y^{25} + x^{26}y^{24} + x^{26}y^{23} + x^{26}y^{21} + x^{26}y^{20} +
    x^{26}y^{18} + x^{26}y^{16} + x^{26}y^{13} + x^{26}y^{11} + x^{26}y^9 + x^{26}y^8 +
    x^{26}y^7 + x^{26}y^6 + x^{25}y^{27} + x^{25}y^{26} + x^{25}y^{25} + x^{25}y^{20} +
    x^{25}y^{16} + x^{25}y^{15} + x^{25}y^{14} + x^{25}y^{13} + x^{25}y^{12} + x^{25}y^9 +
    x^{25}y^8 + x^{25}y^5 + x^{24}y^{27} + x^{24}y^{24} + x^{24}y^{21} + x^{24}y^{17} +
    x^{24}y^{13} + x^{24}y^{12} + x^{24}y^{10} + x^{24}y^8 + x^{24}y^7 + x^{24}y^6 +
    x^{24}y^5 + x^{24}y^4 + x^{23}y^{27} + x^{23}y^{26} + x^{23}y^{24} + x^{23}y^{23} +
    x^{23}y^{21} + x^{23}y^{20} + x^{23}y^{16} + x^{23}y^{15} + x^{23}y^{14} + x^{23}y^{10} +
    x^{23}y^8 + x^{23}y^7 + x^{23}y^6 + x^{23}y^5 + x^{22}y^{27} + x^{22}y^{25} +
    x^{22}y^{24} + x^{22}y^{22} + x^{22}y^{21} + x^{22}y^{17} + x^{22}y^{15} + x^{22}y^{11} +
    x^{22}y^8 + x^{22}y^7 + x^{22}y^6 + x^{22}y^5 + x^{21}y^{27} + x^{21}y^{26} +
    x^{21}y^{25} + x^{21}y^{23} + x^{21}y^{22} + x^{21}y^{20} + x^{21}y^{18} + x^{21}y^{16} +
    x^{21}y^{14} + x^{21}y^{12} + x^{21}y^{10} + x^{21}y^9 + x^{21}y^8 + x^{21}y^6 +
    x^{20}y^{27} + x^{20}y^{24} + x^{20}y^{23} + x^{20}y^{21} + x^{20}y^{20} + x^{20}y^{19} +
    x^{20}y^{15} + x^{20}y^{11} + x^{20}y^{10} + x^{20}y^9 + x^{20}y^8 + x^{20}y^7 +
    x^{20}y^6 + x^{20}y^4 + x^{19}y^{27} + x^{19}y^{26} + x^{19}y^{24} + x^{19}y^{22} +
    x^{19}y^{21} + x^{19}y^{18} + x^{19}y^{17} + x^{19}y^{12} + x^{19}y^{11} + x^{19}y^8 +
    x^{19}y^5 + x^{19}y^3 + x^{18}y^{27} + x^{18}y^{25} + x^{18}y^{24} + x^{18}y^{23} +
    x^{18}y^{22} + x^{18}y^{20} + x^{18}y^{19} + x^{18}y^{18} + x^{18}y^{17} + x^{18}y^{15} +
    x^{18}y^{13} + x^{18}y^{11} + x^{18}y^{10} + x^{18}y^9 + x^{18}y^8 + x^{18}y^6 +
    x^{18}y^3 + x^{18}y^2 + x^{17}y^{27} + x^{17}y^{26} + x^{17}y^{25} + x^{17}y^{21} +
    x^{17}y^{20} + x^{17}y^{18} + x^{17}y^{17} + x^{17}y^{16} + x^{17}y^{14} + x^{17}y^{13} +
    x^{17}y^{11} + x^{17}y^{10} + x^{17}y^9 + x^{17}y^7 + x^{17}y^5 + x^{17}y^4 + x^{17}y^3
    + x^{17}y + x^{16}y^{28} + x^{16}y^{27} + x^{16}y^{26} + x^{16}y^{24} + x^{16}y^{21} +
    x^{16}y^{20} + x^{16}y^{19} + x^{16}y^{18} + x^{16}y^{17} + x^{16}y^{15} + x^{16}y^{14} +
    x^{16}y^{13} + x^{16}y^{11} + x^{16}y^{10} + x^{16}y^9 + x^{16}y^8 + x^{16}y^7 +
    x^{16}y^4 + x^{16}y^2 + x^{16}y + x^{16} + x^{15}y^{27} + x^{15}y^{25} + x^{15}y^{24} +
    x^{15}y^{23} + x^{15}y^{21} + x^{15}y^{19} + x^{15}y^{18} + x^{15}y^{17} + x^{15}y^{15} +
    x^{15}y^{14} + x^{15}y^{12} + x^{15}y^{11} + x^{15}y^{10} + x^{15}y^8 + x^{15}y^7 +
    x^{15}y^3 + x^{15}y^2 + x^{15}y + x^{14}y^{26} + x^{14}y^{25} + x^{14}y^{22} + x^{14}y^{20}
    + x^{14}y^{19} + x^{14}y^{18} + x^{14}y^{17} + x^{14}y^{15} + x^{14}y^{13} + x^{14}y^{11} +
    x^{14}y^{10} + x^{14}y^9 + x^{14}y^8 + x^{14}y^6 + x^{14}y^5 + x^{14}y^4 + x^{14}y^3
    + x^{14}y + x^{13}y^{25} + x^{13}y^{23} + x^{13}y^{20} + x^{13}y^{17} + x^{13}y^{16} +
    x^{13}y^{11} + x^{13}y^{10} + x^{13}y^7 + x^{13}y^6 + x^{13}y^4 + x^{13}y^2 + x^{13}y +
    x^{12}y^{24} + x^{12}y^{22} + x^{12}y^{21} + x^{12}y^{20} + x^{12}y^{19} + x^{12}y^{18} +
    x^{12}y^{17} + x^{12}y^{13} + x^{12}y^9 + x^{12}y^8 + x^{12}y^7 + x^{12}y^5 + x^{12}y^4
    + x^{12}y + x^{11}y^{22} + x^{11}y^{20} + x^{11}y^{19} + x^{11}y^{18} + x^{11}y^{16} +
    x^{11}y^{14} + x^{11}y^{12} + x^{11}y^{10} + x^{11}y^8 + x^{11}y^6 + x^{11}y^5 +
    x^{11}y^3 + x^{11}y^2 + x^{11}y + x^{10}y^{23} + x^{10}y^{22} + x^{10}y^{21} + x^{10}y^{20}
    + x^{10}y^{17} + x^{10}y^{13} + x^{10}y^{11} + x^{10}y^7 + x^{10}y^6 + x^{10}y^4 +
    x^{10}y^3 + x^{10}y + x^9y^{23} + x^9y^{22} + x^9y^{21} + x^9y^{20} + x^9y^{18} +
    x^9y^{14} + x^9y^{13} + x^9y^{12} + x^9y^8 + x^9y^7 + x^9y^5 + x^9y^4 +
    x^9y^2 + x^9y + x^8y^{24} + x^8y^{23} + x^8y^{22} + x^8y^{21} + x^8y^{20} +
    x^8y^{18} + x^8y^{16} + x^8y^{15} + x^8y^{11} + x^8y^7 + x^8y^4 + x^8y +
    x^7y^{23} + x^7y^{20} + x^7y^{19} + x^7y^{16} + x^7y^{15} + x^7y^{14} + x^7y^{13} +
    x^7y^{12} + x^7y^8 + x^7y^3 + x^7y^2 + x^7y + x^6y^{22} + x^6y^{21} +
    x^6y^{20} + x^6y^{19} + x^6y^{17} + x^6y^{15} + x^6y^{12} + x^6y^{10} + x^6y^8 +
    x^6y^7 + x^6y^5 + x^6y^4 + x^6y^3 + x^6y + x^5y^{21} + x^5y^{17} +
    x^5y^{15} + x^5y^{13} + x^5y^{11} + x^5y^9 + x^5y^8 + x^5y^6 + x^5y^5 +
    x^5y^4 + x^5y^2 + x^5y + x^4y^{20} + x^4y^{18} + x^4y^{15} + x^4y^{14} +
    x^4y^{12} + x^4y^{10} + x^4y^9 + x^4y^7 + x^4y^6 + x^4y + x^3y^{17} +
    x^3y^{16} + x^3y^{15} + x^3y^{13} + x^3y^{10} + x^3y^8 + x^3y^7 + x^3y^5 +
    x^3y^2 + x^3y + x^2y^{16} + x^2y^{14} + x^2y^{13} + x^2y^{12} + x^2y^{10} +
    x^2y^9 + x^2y^8 + x^2y^6 + x^2y^5 + x^2y^4 + x^2y^2 + x^2y + xy^{15} +
    xy^{14} + xy^{11} + xy^{10} + xy^7 + xy^6 + xy^3 + xy^2 + y^{14} + y^{12} +
    y^{10} + y^8 + y^6 + y^4 + y^2 + 1). $  

\vspace{0.6\baselineskip}
The following polynomial is the resultant of  Eq. \eqref{9.9} in Theorem \ref{th:2.9}.
\vspace{0.6\baselineskip}

$ (x^{83} + x^{80} + x^{79} + x^{78} + x^{76} + x^{74} + x^{72} + x^{69} + x^{68} + x^{67} + x^{64} + x^{63} + x^{59} + x^{58} + x^{57} + x^{55} + x^{54} + x^{53} + x^{51} + x^{50} + x^{49} + x^{48} + x^{46} + x^{44} + x^{42} + x^{39} + x^{37} + x^{35} + x^{33} + x^{32} + x^{28} + x^{24} + x^{23} + x^{22} + x^{21} + x^{20} + x^{18} + x^{17} + x^{16} + x^{13} + x^{12} + x^{11} + x^{10} + x^8 + x^6 + x^4 + x^3 + x^2 + 1)^2
 (x^{83} + x^{81} + x^{76} + x^{75} + x^{74} + x^{73} + x^{{71}} + x^{69} + x^{68} + x^{67} + x^{63} + x^{59} + x^{58} + x^{57} + x^{55} + x^{54} + x^{53} + x^{52} + x^{47} + x^{46} + x^{45} + x^{44} + x^{37} + x^{36} + x^{35} + x^{34} + x^{31} + x^{30} + x^{29} + x^{27} + x^{26} + x^{25} + x^{20} + x^{19} + x^{18} + x^{17} + x^{16} + x^{14} + x^{13} + x^{11} + x^5 + x^4 + x^3 + x + 1)^2
 (x^{83} + x^{81} + x^{80} + x^{79} + x^{77} + x^{75} + x^{73} + x^{72} + x^{{71}} + x^{70} + x^{67} + x^{66} + x^{65} + x^{63} + x^{62} + x^{61} + x^{60} + x^{59} + x^{55} + x^{51} + x^{50} + x^{48} + x^{46} + x^{44} + x^{41} + x^{39} + x^{37} + x^{35} + x^{34} + x^{33} + x^{32} + x^{30} + x^{29} + x^{28} + x^{26} + x^{25} + x^{24} + x^{20} + x^{19} + x^{16} + x^{15} + x^{14} + x^{11} + x^9 + x^7 + x^5 + x^4 + x^3 + 1)^ 2
 (x^{83} + x^{82} + x^{76} + x^{75} + x^{{71}} + x^{70} + x^{66} + x^{65} + x^{63} + x^{62} + x^{61} + x^{60} + x^{58} + x^{57} + x^{54} + x^{53} + x^{52} + x^{51} + x^{50} + x^{49} + x^{46} + x^{45} + x^{44} + x^{43} + x^{42} + x^{41} + x^{39} + x^{38} + x^{36} + x^{35} + x^{34} + x^{33} + x^{30} + x^{29} + x^{26} + x^{25} + x^{23} + x^{22} + x^{21} + x^{20} + x^{19} + x^{17} + x^{16} + x^{15} + x^{12} + x^{11} + x^{10} + x^9 + x^8 + x^7 + x^6 + x^4 + x^2 + x + 1)^2
 (x^{83} + x^{82} + x^{80} + x^{79} + x^{78} + x^{72} + x^{70} + x^{69} + x^{67} + x^{66} + x^{65} + x^{64} + x^{63} + x^{58} + x^{57} + x^{56} + x^{54} + x^{53} + x^{52} + x^{49} + x^{48} + x^{47} + x^{46} + x^{39} + x^{38} + x^{37} + x^{36} + x^{31} + x^{30} + x^{29} + x^{28} + x^{26} + x^{25} + x^{24} + x^{20} + x^{16} + x^{15} + x^{14} + x^{12} + x^{10} + x^9 + x^8 + x^7 + x^2 + 1)^ 2
 (x^{83} + x^{82} + x^{81} + x^{79} + x^{77} + x^{76} + x^{75} + x^{74} + x^{73} + x^{72} + x^{{71}} + x^{68} + x^{67} + x^{66} + x^{64} + x^{63} + x^{62} + x^{61} + x^{60} + x^{58} + x^{57} + x^{54} + x^{53} + x^{50} + x^{49} + x^{48} + x^{47} + x^{45} + x^{44} + x^{42} + x^{41} + x^{40} + x^{39} + x^{38} + x^{37} + x^{34} + x^{33} + x^{32} + x^{31} + x^{30} + x^{29} + x^{26} + x^{25} + x^{23} + x^{22} + x^{21} + x^{20} + x^{18} + x^{17} + x^{13} + x^{12} + x^8 + x^7 + x + 1)^2.$



\end{document}